\documentclass[12pt]{article}
\usepackage[utf8]{inputenc}
\usepackage[margin=1in]{geometry}
\usepackage{amsmath}
\usepackage{amssymb}
\usepackage{amsthm}
\usepackage{natbib}
\PassOptionsToPackage{hyphens}{url} 
\usepackage{hyperref}
\usepackage{graphicx}
\usepackage{booktabs}
\usepackage{tabularx}
\usepackage{ragged2e} 
\usepackage{array}    
\usepackage{longtable} 
\newcolumntype{L}{>{\raggedright\arraybackslash}X}       
\newcolumntype{P}[1]{>{\raggedright\arraybackslash}p{#1}} 
\newcolumntype{Q}[1]{>{\fontsize{9}{11}\selectfont\raggedright\arraybackslash}p{#1}} 

\newtheorem{theorem}{Theorem}

\newtheorem{proposition}[theorem]{Proposition}
\newtheorem{corollary}[theorem]{Corollary}

\theoremstyle{definition}
\newtheorem{definition}[theorem]{Definition}
\newtheorem{example}[theorem]{Example}
\newtheorem*{remark}{Remark}
\makeatletter
\newcommand{\keepwithlist}{\@beginparpenalty\@M}
\makeatother
\newcommand{\paragraphit}[1]{\paragraph*{\textnormal{\textit{#1}}}}

\title{The Tragedy of Productivity: A Unified Framework for Diagnosing
  Coordination Failures in Labor Markets and AI Governance\thanks{This
    is a preliminary working paper. This version corrects and extends
    the main theorem of earlier versions, corrects the accounting of
    the productivity model, the arithmetic and interpretation of the
    Tragedy Index, and several factual statements and citations, and
    adds an exact condition for the productivity model and a
    proposition on sustaining cooperation in repeated play. Three
    separate papers expanding the theory, the productivity analysis
    and the AI governance analysis are in preparation. We specifically
    welcome feedback and critical engagement from field experts. The
    views expressed are those of the author and do not necessarily
    represent the views of any affiliated organizations.}}

\author{Ali Dasdan\thanks{This article was written with AI assistance
    (Claude, Codex, Gemini), reflexively demonstrating the
    structural pressure of adoption for efficiency while using the
    technology to analyze its own governance. All arguments,
    interpretations, and conclusions are the author's own, as is
    responsibility for any errors.}  \\ KD Consulting\\ Saratoga, CA,
  USA\\ alidasdan@gmail.com}

\date{November 26, 2026}

\begin{document}

\maketitle

\begin{abstract}
  In 1930, John Maynard Keynes predicted that technological progress
  would reduce the working week to 15 hours within a hundred
  years. Despite productivity having grown severalfold since then,
  working hours remain far above that benchmark. Separately,
  leading AI researchers warn of existential risks from artificial
  intelligence, yet development accelerates rather than slows. We
  argue these failures share a common game-theoretic structure:
  coordination
  failures arising when individual rationality produces collective
  harm.

  We synthesize from established game theory five conditions that we
  use to diagnose such coordination failures as \textit{structural
  tragedies}: N-player structure with multiple
  independent actors, binary choices generating negative
  externalities, dominance property where defection yields higher
  payoffs, Pareto-inefficiency where cooperation dominates mutual
  defection, and enforcement difficulty from structural barriers. The
  dominance and Pareto-inefficiency conditions are sufficient for a
  unique, Pareto-dominated equilibrium of the stage game and, among
  identical actors, also necessary; for repeated play we give an
  exact condition under which a simple trigger arrangement sustains
  cooperation. We extend this framework to grade tragedy severity through
  condition intensities, introducing a heuristic Tragedy Index under
  which governance of transformative AI breakthroughs receives the
  highest composite score among the cases assessed, above climate
  change, nuclear weapons, and bank runs. We illustrate this framework with
  canonical cases including commons tragedies, security dilemmas, bank
  runs, and space debris.

  Applied to \textit{productivity} competition, we give in a
  simplified model a necessary and sufficient condition, one
  inequality per firm in its hours and market share, for universal
  output expansion to be the only equilibrium although firms and
  workers would all be better off under universal restraint. This
  coordination failure prevents productivity gains from translating
  fully to worker welfare. We show this formally for working hours; we
  conjecture that the same structural logic extends to wages,
  benefits, and job security. European evidence is consistent with
  productivity-welfare decoupling persisting even under favorable
  institutional conditions, manifesting differently across
  institutional contexts.

  Applied to \textit{AI governance}, we argue that development has the
  same structure when its payoff conditions hold, and we assess it
  qualitatively as less favorable for
  coordination than successful arms control cases on most of eight
  dimensions. The Russia-Ukraine drone war illustrates the competitive
  pressure involved: despite international dialogue on weapons
  governance, both sides rapidly escalated drone production and use.

  The analysis is diagnostic rather than prescriptive: we synthesize
  from the literature structural barriers both to standard solutions
  and to the fundamental changes that would be required to resolve
  these tragedies, rather than proposing implementable solutions whose
  feasibility our framework questions.
\end{abstract}

\tableofcontents

\section{Introduction}

\subsection{The Productivity Puzzle}

John Maynard Keynes famously predicted in 1930 that technological
progress would reduce the working week to 15 hours within a hundred
years\footnote{While Keynes's phrasing was speculative (``Three-hour
shifts or a fifteen-hour week may put off the problem''), the
literature has standardly interpreted this as a prediction of future
working hours (e.g., \citealp{crafts2022}).} \citep{keynes1930}. With
2030 approaching, hours remain far above fifteen per week. In the
countries covered by the OECD series, reported average annual hours
per worker in 2024 ranged from about 1,330 to over 2,200
\citep{oecd2025}, well above the 780 hours that correspond to fifteen
hours per week, and in many sectors
working hours have increased rather than decreased \citep{schor1992};
this competitive escalation is exemplified by the `996' work culture
(9 am to 9 pm, 6 days a week) in China \citep{wiki996}, which has even
been emulated by some US startups seeking to match that intensity
\citep{nytimes2025hustle}.

This failure demands explanation. The technological capacity for
reduced working hours clearly exists: productivity per worker-hour has
increased severalfold since Keynes wrote. In the United States, for
example, output per hour rose about sixfold between 1930 and 2007 (our
calculation from the interval growth rates in Table 4 of
\citet{gordon2010}), while annual hours per worker fell by about a fifth between
1929 and 2000 \citep{huberman2007}. The material abundance Keynes envisioned has been
achieved: modern economies produce far more goods and services per
hour of labor than was possible in 1930. Yet this capacity for leisure
remains unrealized.

Why? Conventional explanations typically invoke cultural factors or
policy failures. We argue these explanations are incomplete. The
persistence of long working hours despite dramatic productivity gains
reflects not mere policy failure or cultural preferences, but a
structural coordination problem analogous to tragedies in other
domains.

This study is also motivated by a persistent empirical pattern
observed by the author in the software and technology sectors over the
last two decades. In these industries, leadership consistently demands
output expansion following productivity gains, often reducing
engineering headcount rather than stabilizing it. This `do more with
less' dynamic persists even when cost-cutting threatens organizational
sustainability.

\subsection{The AI Governance Puzzle}

A parallel puzzle emerges in artificial intelligence
development. Leading AI researchers consistently warn of existential
risks from advanced AI systems \citep{bostrom2014,
  russell2015}. Survey evidence shows substantial minorities of AI
researchers estimate non-trivial probabilities of extremely bad
outcomes, including human extinction \citep{bengio2025,
  grace2018}. Researchers have called for cooperation among AI
companies to support safe and responsible development
\citep{askell2019}.

Yet development accelerates rather than slows. Investment in AI
capabilities races ahead \citep{fedreserve2025}. Voluntary pauses fail
to materialize. Safety measures, while present, appear insufficient
relative to the risks researchers identify
\citep{askell2019}. International coordination remains nascent and
non-binding.

Why can we not coordinate even when the stakes appear existential? Why
do warnings from the field's leading researchers fail to slow
development? Why do calls for safe development coexist with accelerating
capabilities races?

\subsection{The Thesis: Structural Identity}

We argue these puzzles exhibit the same mathematical structure. Both
represent \textit{structural tragedies}, coordination failures where
individually rational choices produce collectively suboptimal outcomes
and where coordination faces overwhelming structural barriers.

If this thesis is right, it has important implications. First, it
offers an explanation of why standard solutions struggle in both
domains: the barriers are structural, not merely political or
cultural. Second, it supports conditional expectations about which
coordination attempts are likely to succeed or fail, based on whether
they address the underlying structure. Third, it suggests what kinds
of interventions might succeed: those that alter the competitive
architecture itself.

\subsection{Preview of Contributions}

This paper makes four contributions.

\textbf{First}, we synthesize from established game theory
\citep{dawes1980, jervis1978} a framework of five conditions that we
use to diagnose structural tragedies; two of them (dominance and
Pareto-inefficiency) are sufficient for a unique, Pareto-dominated
equilibrium of the stage game and, among identical actors, necessary
for it. We illustrate the framework with canonical cases including
environmental commons \citep{hardin1968}, security dilemmas
\citep{mearsheimer2001} and bank runs \citep{diamond1983} (Section
\ref{sec:framework}).

\textbf{Second}, we grade each condition by intensity and introduce a
heuristic Tragedy Index that compares assessed condition intensities
across domains; under this rubric governance of transformative AI
breakthroughs receives the highest composite score among the cases
assessed, above climate change, nuclear weapons, and bank runs
(Section \ref{sec:tragedy-index}).

\textbf{Third}, we apply the framework to productivity competition in
a stylized model, that is, a deliberately simplified one with two
actions per firm, a fixed market revenue and fixed pay. We derive an
exact condition, one inequality per firm, under which universal
expansion of output is the unique, Pareto-dominated equilibrium. This
offers a structural explanation of why Keynes's expectation has not
been realized. European evidence is consistent with the view that,
even under favorable institutional conditions, coordination achieves
partial and costly overrides of the tragic logic that remain under
pressure (Section \ref{sec:productivity}).

\textbf{Fourth}, we argue that AI development exhibits the same
structure when stated payoff inequalities hold, with higher assessed
intensity. A qualitative comparison with historical arms control
suggests that AI faces less favorable conditions on most dimensions
where nuclear, biological, chemical (NBC) weapons coordination
succeeded, and the Russia-Ukraine drone war illustrates how rapidly
competitive dynamics can override governance intentions (Section
\ref{sec:ai}).

\paragraph{Machine-checked proofs:} Every numbered result of this
paper (Theorem \ref{thm:main}, Propositions \ref{prop:trigger},
\ref{prop:dominance} and \ref{prop:pareto-productivity}, Example
\ref{ex:two-firm} and Corollary \ref{cor:equal-productivity}),
together with the arithmetic of the Tragedy Index in Table
\ref{tab:intensity-comparison}, has been formally verified in Lean 4
and, in a separate development and with the qualification that
follows, in Dafny; both check each stated equivalence in both
directions. The Lean proofs construct the real numbers and use only
Lean's standard logical foundations, with no added axioms; the Dafny
proofs take as given two textbook facts of real analysis that they do
not prove, the existence of the value of a bounded discounted sum and
of fifth roots. The proof scripts are in the \texttt{dafny} and
\texttt{lean} directories of the accompanying public repository
\citep{dasdan2026proofs}; they refer to Table
\ref{tab:intensity-comparison} as the Index table. The formal statements were written to match
the statements in the text; that correspondence was checked by
reading, not by machine.

Section \ref{sec:literature} reviews the related literature, and
Section \ref{sec:implications} discusses predictions, falsifiability
and limitations.

\section{Related Literature}
\label{sec:literature}

Our work synthesizes insights from multiple, often separate, literatures.

\subsection{The Prisoner's Dilemma and Social Dilemmas}

The prisoner's dilemma, studied at length by
\citet{rapoport1965}, is the canonical coordination
failure. Two players each face a choice between cooperation and
defection. Defection dominates cooperation regardless of the other's
choice, yet mutual cooperation Pareto-dominates mutual defection. The
unique Nash equilibrium—mutual defection—is recognizably suboptimal.

\citet{schelling1960, schelling1973} extended this analysis to
N-player settings with binary choices and externalities. When multiple
actors\footnote{We use `player' and `actor' interchangeably for game
participants.} face similar decisions where individual choices impose
externalities on others, coordination failures can emerge even when
all actors understand the collective problem. A key insight is that,
in models where the private gain from defecting and the total harm it
causes are fixed and each actor bears only a share of order $1/N$ of
that harm, the individual incentive to deviate grows stronger as N
increases—each actor's restraint provides smaller individual benefit
while the temptation to defect remains constant.

\citet{dawes1980} formalized the structure of ``social dilemmas,''
defining them as N-person games where individual rationality leads to
collective irrationality. His key contribution was identifying the
mathematical structure common to diverse coordination failures:
defection dominates cooperation for each individual, yet universal
cooperation yields higher payoffs than universal defection. These two
inequalities define the dilemma for a single play of the game. From
the experimental evidence he reviewed, Dawes concluded that
cooperation is helped when people understand the dilemma and have
reason to believe that others will not defect, naming knowledge,
morality and trust as its ingredients. Whether the outcome persists
under repeated interaction or institutional intervention is a separate
question, on which he gives no formal result. \citet{wu1990} give a
linear N-person family in which the prisoner's dilemma, Chicken and
trust dilemmas are regions of the payoff parameters, together with a
``structural cooperation index'' computed from those parameters: an
early precedent for classifying dilemmas by payoff inequalities and
for payoff-based indices.

\citet{jervis1978} applied this logic to international relations,
analyzing the ``security dilemma'' (originally coined by
\citet{herz1950}) where states' efforts to increase their own
security—through arms buildups or territorial expansion—threaten other
states and trigger countermeasures that leave all parties less
secure. In its prisoner's-dilemma form, the security dilemma exhibits
the core social dilemma structure: each state has incentive to arm
regardless of others' choices (dominance), yet mutual restraint would
benefit all (Condition 4), but international anarchy prevents
enforceable agreements. Jervis also analyzes the problem as a Stag
Hunt, in which a state prefers restraint when it expects the others to
restrain, so dominance does not hold in every case. He argued that the
dynamic can persist even when all states prefer peace, because the
anarchic structure of international relations makes credible
commitment to cooperation difficult.

The folk theorem literature \citep{fudenberg1986} proves that repeated
interaction can sustain cooperation under certain conditions, for
example through trigger strategies, when players are sufficiently
patient and deviations are detected and punished well enough. Perfect
observation of actions is not required in general
\citep{greenporter1984}. In this paper, Proposition
\ref{prop:trigger} treats a specified trigger arrangement and, when a
defector cannot be held below its payoff under universal defection,
characterizes permanent full cooperation under two monitoring
structures; false alarms, outside sanctions and general institutional
design remain open. Whether cooperation can be sustained over time in
productivity competition and AI development is treated below as a
hypothesis, not as a result.

\subsection{Tragedies in Specific Domains}

Several literatures identified the same structure domain by domain:
resource depletion where property rights are absent
\citep{hardin1968}, arms races under the security dilemma
\citep{jervis1978} and their extension to offensive realism
\citep{mearsheimer2001}, the undersupply of public goods
\citep{olson1965}, and bank runs \citep{diamond1983}.
\citet{ostrom1990} showed that under certain conditions—small groups,
clear boundaries, monitoring capacity, graduated
sanctions—communities can overcome commons tragedies through
institutional design; many commons lack these conditions. Section
\ref{sec:canonical} assesses these cases against the five conditions,
treating them as variations of the same social dilemma.

\subsection{Productivity and Working Hours}

The literature on working hours exhibits puzzlement about why
productivity gains do not translate into leisure \citep{keynes1930}.

Historical research documents the trajectory of working hours: annual
hours in industrialized economies declined from over 3,000 in 1870 to
between roughly 1,350 and 1,900 by 2000 \citep{huberman2007}, with US
and European paths diverging after 1980 (Section
\ref{sec:europe}). \citet{pencavel2018} documents the relationship
between productivity and hours, finding that after initial declines,
working hours have remained stubbornly stable.

Theoretical explanations focus primarily on preferences and
culture. \citet{becker1965} developed time allocation theory showing
how rising wages could increase or decrease labor supply depending on
income and substitution effects. \citet{schnaiberg1980} introduced the
``treadmill of production'' concept, arguing that competitive
pressures force continuous expansion. \citet{alesina2005} examine
US-Europe differences in working hours and conclude that European
labor market regulations, advocated by unions in declining industries
who argued ``work less, work all,'' explain the bulk of the
difference, rather than cultural preferences or tax structures alone.

Recent work by Acemoglu and Restrepo \citep{acemoglu2019automation,
  acemoglu2020} and by \citet{acemoglu2025} provides important
insights into technology adoption. \citet{acemoglu2019automation}
distinguish automation that displaces labor, including ``so-so''
technologies whose small productivity gains do not offset the
displacement, from technologies that create new tasks in which labor
has a comparative advantage, and they point to incentives for
excessive automation.

These theories explain specific drivers: competitive pressure
\citep{schnaiberg1980} and technology choice
\citep{acemoglu2019automation, acemoglu2025}. We distinguish between
\textit{which} technologies firms adopt (Acemoglu's focus) and
\textit{how} productivity gains are utilized (our focus). Extensive
research also documents the secular decline in the labor share of
income globally \citep{karabarbounis2014global, elsby2013decline}. To
our knowledge the literature lacks a game-theoretic model of the
coordination problem in \textit{allocation}: of why competitive
pressure can prevent even human-augmenting productivity gains from
translating to worker welfare. We provide one. In a stylized N-player
model and under an exact condition, universal expansion of output is
the only equilibrium, yet universal restraint would Pareto-dominate it
for the modeled firms and their workers; we hypothesize that barriers
to coordination make this outcome persist. The model offers an
explanation of \citet{schnaiberg1980}'s treadmill, complements
\citet{acemoglu2019automation, acemoglu2025}'s analysis of adoption
failures with a theory of allocation failures, and suggests why the
regulatory interventions documented by \citet{alesina2005} require
costly, persistent institutional support to partially override market
dynamics.

\subsection{AI Governance and Safety}

The AI safety literature has grown rapidly in recent years,
identifying multiple challenges in developing safe and beneficial AI
systems.

\citet{bostrom2014} provided a comprehensive analysis of existential
risks from superintelligent AI, emphasizing the orthogonality thesis
(intelligence and goals are independent) and instrumental convergence
(advanced AI systems will pursue convergent subgoals regardless of
final objectives). \citet{russell2019} framed AI safety as an inverse
reinforcement learning problem and proposed provably beneficial AI
through value alignment. \citet{amodei2016} identified concrete
problems in AI safety: avoiding negative side effects, avoiding reward
hacking, scalable supervision, safe exploration, and robustness to
distributional shift. \citet{dafoe2018} analyzed AI governance
challenges, identifying the need for international coordination on
development and deployment. \citet{bengio2025} provided a
comprehensive synthesis of the current evidence on the capabilities,
risks, and safety of advanced AI systems. \citet{panigrahy2025} derive an incompatibility under strict
mathematical definitions that they propose: with safety defined as
never making a false claim, a safe and trusted AI system cannot match
human capability on every task, which is their definition of
artificial general intelligence (AGI). \citet{armstrong2016} analyzed ``racing to
the precipice'', the concern that competitive pressure could drive
unsafe AI development. Their formal model shows how competition can
lead development teams to skimp on safety; it does not address the
broader set of conditions we consider here.

While this literature thoroughly documents AI risks and governance
challenges and includes formal race models, a structural comparison
with other coordination failures is less developed.

\section{A Unified Framework for Structural Tragedies}
\label{sec:framework}

We now present the framework: five conditions, a theorem on the stage
game, a proposition on repeated play, and canonical cases. It is an
initial modeling framework with explicit simplifying assumptions.

We use the term ``Tragedy'' in the strict sense defined by
\citet{whitehead1925} and popularized by \citet{hardin1968}: ``the
solemnity of the remorseless working of things.'' A situation becomes
a tragedy not through error or malice but through its structure:
individual rationality favors defection even though universal
cooperation would be better for every actor, while structural barriers
stand in the way of coordination. Actors are not making mistakes; they
are constrained by the logic of their own survival. In this paper the
remorselessness is a hypothesis about persistence, stated after
Theorem \ref{thm:main} and made precise for one class of arrangements
in Proposition \ref{prop:trigger}; it is not a conclusion of that
theorem.

We formalize the payoff structure as follows. Let $s = (s_1, s_2,
\ldots, s_n)$ be a strategy profile for $n$ players, which is a
vector specifying the chosen strategy $s_i \in \{C, D\}$ for every
player $i$, where $C$ represents Cooperate (restraint) and $D$
represents Defect (advantage-seeking). We write $n$ for the number of
players throughout; the capital N in ``N-player'' and ``N-person'',
which follows the literature, denotes the same number.

The utility (or payoff) for a specific player $i$ is given by a
function $U_i$ that maps this complete strategy profile to a
real-valued payoff:
\begin{equation}\label{eq:payoff}
U_i(s) = U_i(s_1, s_2, \ldots, s_n)
\end{equation}

For clarity, we will often separate the strategy of player $i$ from
the strategies of all other players. We denote the strategy profile of
all players except $i$ as $s_{-i}$. This allows us to write
the utility function as $U_i(s_i, s_{-i})$.

Using this formal notation, we can precisely define the five
conditions for the tragedy.

\subsection{The Five Conditions}

We call a problem a \textit{structural tragedy} when two things
hold. First, universal defection is the only equilibrium of the
one-shot game although every player would be better off under
universal cooperation. Second, structural obstacles to enforcement, of
the kinds listed in Condition 5, stand in the way of agreements to
cooperate. The first part is a property of the stage game, which
Section \ref{sec:formal-result} makes precise. The second is stated
informally. Whether cooperation in fact cannot be sustained is a
further claim, which must be made relative to a stated class of
arrangements; in our applications it is a hypothesis. The following
five conditions are what we check in an application. Conditions 3 and
4 are sufficient for the first part and, among identical players, also
necessary (Theorem \ref{thm:main}).

\begin{definition}[C1: N-Player Structure]
Multiple independent decision-makers ($n \geq 2$) each control their
own choices. No single player can solve the problem
unilaterally. Choices are made in a decentralized architecture without
centralized coordination.
\end{definition}

\begin{definition}[C2: Binary Choice with Externalities]
Each player $i$ faces a binary choice between two strategies: Cooperate
($C$) representing restraint, or Defect ($D$) representing
advantage-seeking. Individual choices create externalities affecting
all players. Specifically, one player's defection imposes negative
externalities on others.
\end{definition}

The binary structure simplifies the game while capturing the essential
strategic tension. In practice, choices may appear continuous, but the
strategic logic often reduces to a binary decision: cooperate or
defect. The key is that individual choices impose externalities: one
player's defection lowers others' payoffs. This does not necessarily
increase the affected actor's relative incentive to defect: strategic
complementarity is a separate property, which we do not assume here.

\begin{definition}[C3: Dominance Property]
Defection yields higher payoff than cooperation regardless of other
players' choices. Formally,
\begin{equation}\label{eq:c3}
U_i(D, s_{-i}) > U_i(C, s_{-i}) \quad \text{for any possible } s_{-i}
\end{equation}
\end{definition}

Dominance creates the dilemma: each player has a strict incentive to
defect regardless of what others do. It is often driven by existential
stakes, as when cooperating while others defect leads to elimination
(e.g., death, bankruptcy).

\begin{definition}[C4: Pareto-Inefficiency]
Universal cooperation Pareto-dominates universal defection. Formally,
let $s^C = (C, C, \ldots, C)$ be the Universal Cooperation
profile and let $s^D = (D, D, \ldots, D)$ be the Universal
Defection profile. This condition states that every player strictly
prefers the universal cooperation outcome:
\begin{equation}\label{eq:c4}
U_i(s^C) > U_i(s^D) \quad \text{for all players } i.
\end{equation}
\end{definition}

Under this condition all players would prefer mutual cooperation to
mutual defection if they could coordinate. The condition compares
payoffs as represented in the model. It does not by itself establish
that actors know or recognize the collective advantage of cooperation;
such recognition requires an additional informational assumption. When
actors do know these payoffs, the tragedy is visible, predictable, and
collectively unwanted, yet the structure makes it difficult to escape.

\begin{definition}[C5: Enforcement Difficulty]
Structural barriers prevent effective coordination mechanisms. These
barriers may include: (i) anarchy as in no central authority with
enforcement power, (ii) verification difficulty as in actions
cannot be reliably monitored, (iii) monitoring infeasibility as in
costs of monitoring exceed benefits of coordination, or (iv)
commitment problems as in players cannot credibly commit to future
cooperation.
\end{definition}

This condition distinguishes structural tragedies from simple
coordination problems. Many Pareto-inefficient equilibria can be
overcome through contracts, monitoring, or institutional
design. We hypothesize that structural tragedies persist because
enforcement mechanisms face fundamental barriers: not merely high
costs but severe structural obstacles. This condition is stated
informally and is not used in the proof of Theorem
\ref{thm:main}. A formal counterpart requires a stated class of
arrangements available to the actors. For one such class, agreements
enforced only by the threat of returning to universal defection when a
departure is detected, Proposition \ref{prop:trigger} gives an exact
condition. A treatment with a specified set of institutions is left to
future work.

\subsection{The Formal Result}
\label{sec:formal-result}

We now state what Conditions 3 and 4 imply for a single play of the
game. Two quantities summarize each player's position. The \emph{gain
from defecting} against a profile $s_{-i}$ of the others is
\begin{equation}\label{eq:gain}
  g_i(s_{-i}) = U_i(D, s_{-i}) - U_i(C, s_{-i}).
\end{equation}
We write $G_i = g_i(s^C_{-i})$ for the gain when all others cooperate;
it may have either sign. The \emph{surplus from cooperation} is $v_i =
U_i(s^C) - U_i(s^D)$. Condition 3 says that $g_i > 0$ against every
profile, and Condition 4 that $v_i > 0$. Equilibria include mixed
strategies, compared by expected payoffs, and ``Pareto-dominated''
means strictly worse for every player, as in Condition 4.

We distinguish three strengths of claim.

\begin{definition}[Three Levels of Claim]\label{def:levels}
\emph{Level 1, a dominated equilibrium exists:} $s^D$ is a Nash
equilibrium and $v_i > 0$ for every player.
\emph{Level 2, a unique inefficient equilibrium:} $s^D$ is the only
Nash equilibrium and $v_i > 0$ for every player.
\emph{Level 3, failure to sustain full cooperation relative to a
stated class of arrangements:} no arrangement in the class makes
permanent play of $s^C$ self-enforcing.
\end{definition}

Level 2 is a statement about equilibria, not a prediction that
behavior must follow it. Level 3 is always stated with its class of
arrangements, and it does not exclude other outcomes that improve on
$s^D$. The theorem below concerns levels 1 and 2. Its parts (a) to (c)
are standard and immediate from the definitions \citep{dawes1980};
part (d) is elementary. The binary action set defines the setting;
negative externalities are not used in the proof. Condition 5 is not
used either; its role is discussed after the proof.

\begin{theorem}[The Tragedy Conditions and the Stage Game]\label{thm:main}
Consider an N-player game satisfying Conditions 1--2 (structure and
externalities).
\begin{enumerate}
    \item[(a)] $s^D$ is a Nash equilibrium if and only if
      $g_i(s^D_{-i}) \ge 0$ for every player, and $s^C$ is a Nash
      equilibrium if and only if $G_i \le 0$ for every player.
    \item[(b)] If Conditions 3--4 (dominance and Pareto-inefficiency)
      hold, then universal defection $s^D$ is the unique Nash
      equilibrium, in pure or mixed strategies; it is strict, hence
      stable against unilateral deviations; and it is Pareto-dominated
      by universal cooperation $s^C$. That is, level 2 holds.
    \item[(c)] The converse of (b) is false in general: level 2 can
      hold in a game in which defection is not dominant for every
      player.
    \item[(d)] Suppose the players are identical, that is, the game
      is symmetric and anonymous: there are two functions $u_C$ and
      $u_D$, the same for all players, such that a player who
      cooperates receives $u_C(k)$ and a player who defects receives
      $u_D(k)$ when $k$ of the other players cooperate. If $s^D$ is
      the unique Nash equilibrium in pure strategies, then Condition 3
      holds. Hence, for identical players, level 2 holds if and only
      if Conditions 3 and 4 hold.
\end{enumerate}
\end{theorem}

\begin{proof}
We prove the four parts in turn.

\paragraphit{Part (a).}

At $s^D$ each player has one alternative action, $C$, and it is not
profitable exactly when $g_i(s^D_{-i}) \ge 0$; a mixed deviation
cannot do better than both pure actions. The statement for $s^C$ is
the same with the roles of the actions exchanged.

\paragraphit{Part (b).}

Under Condition 3, $D$ is strictly better than $C$ for every player
against every pure profile of the others, hence against every
probability distribution over such profiles. In any equilibrium, in
pure or mixed strategies, every player therefore plays $D$ with
probability one, since a player who put positive probability on $C$
would gain by moving it to $D$; so $s^D$ is the only candidate. It is
an equilibrium,
and a strict one, because a player who changes alone to $C$, with any
positive probability, strictly lowers its own payoff. By Condition 4
every player is strictly better off at $s^C$. This is level 2.

\paragraphit{Part (c).}

Consider the two-player game with payoffs $(3,3)$ if both cooperate,
$(0,4)$ if only the column player defects, $(2,0)$ if only the row
player defects, and $(1,1)$ if both defect. Defection is dominant for
the column player but not for the row player, who prefers to cooperate
against a cooperator. Yet $s^D$ is the unique equilibrium, in pure or
mixed strategies, and it is Pareto-dominated by $s^C$.

\paragraphit{Part (d).}

Let $d(k) = u_D(k) - u_C(k)$ for $k = 0, \ldots, n-1$; because the
functions are common to all players, $d(k)$ is every player's gain
from defecting when $k$ others cooperate. Since $s^C$ is not an
equilibrium, a player gains by defecting when all $n-1$ others
cooperate, so $d(n-1) > 0$. Suppose $d(j) > 0$ for all $j \ge m$,
where $1 \le m \le n-1$, and consider a profile with exactly $m$
cooperators. A defector there faces $m$ cooperators and does not gain
by cooperating, because $d(m) > 0$. The profile is not an equilibrium,
so a cooperator gains by defecting; a cooperator faces $m-1$ other
cooperators, so $d(m-1) > 0$. By induction $d(k) > 0$ for every $k$,
which is Condition 3. The last statement follows with part (b).
\end{proof}

For identical players, Conditions 3 and 4 are the two inequalities by
which \citet{dawes1980} defines an N-person dilemma. Part (d) adds
that among identical players they are also necessary for the outcome.

\paragraph{Sustaining cooperation over time.} Theorem \ref{thm:main}
concerns a single play of the game. When the game is repeated,
cooperation can sometimes be sustained by the threat of returning to
defection. Whether it can depends on the arrangements available, so a
claim that it cannot (level 3) has to name its class of
arrangements. The following proposition treats the simplest class.

Suppose $s^D$ is a Nash equilibrium of the stage game and $v_i > 0$
for every player. Write $R_i = U_i(s^C)$ and $P_i = U_i(s^D)$, so that
$v_i = R_i - P_i$. The game is repeated in periods $0, 1, 2, \ldots$
with a common discount factor $\delta$, $0 \le \delta < 1$, and
payoffs are discounted averages. After the actions of each period and
before those of the next, a public signal with two values, alarm or
no alarm, is realized. Its
distribution depends only on the actions of that period, in the same
way after every history: it raises no alarm when all cooperate, and it
raises an alarm with probability $\psi_i$, $0 < \psi_i \le 1$, when
player $i$ alone departs from $s^C$. There are no false alarms and no
sanctions other than the continuation of play. The \emph{trigger
arrangement} is: cooperate as long as no alarm has occurred, and
defect forever after the first alarm.

\begin{proposition}[Trigger Arrangements]
\keepwithlist
In the repeated game just described:\label{prop:trigger}
\begin{enumerate}
    \item[(i)] The trigger arrangement is a perfect public equilibrium
      if and only if $(1-\delta)\,G_i \le \delta\,\psi_i\,v_i$ for
      every player $i$; equivalently, $\tau_i \le
      \psi_i\,\delta/(1-\delta)$, where $\tau_i = G_i/v_i$.
    \item[(ii)] Suppose in addition that $U_i(D, s_{-i}) \ge P_i$ for
      every player and every $s_{-i}$. Then a perfect public
      equilibrium in which $s^C$ is played in every period exists if
      and only if the inequality of part (i) holds for every player.
    \item[(iii)] Suppose instead that the action profile itself is
      publicly observed after every period (perfect monitoring), and
      that $U_i(D, s_{-i}) \ge P_i$ for every player and every
      $s_{-i}$. Then a subgame-perfect equilibrium in which $s^C$ is
      played in every period exists if and only if $(1-\delta)\,G_i
      \le \delta\,v_i$ for every player $i$.
\end{enumerate}
\end{proposition}

\begin{proof}
(i) In the cooperative phase player $i$'s value is $R_i$. A one-period
departure gives $(1-\delta)(R_i + G_i) + \delta[(1-\psi_i) R_i +
  \psi_i P_i]$, which is at most $R_i$ if and only if $(1-\delta)\,G_i
\le \delta\,\psi_i\,v_i$. After an alarm the arrangement repeats a
stage-game equilibrium whatever is done, so no departure is
profitable. Payoffs are bounded and discounted, and the other
players' strategies depend only on the public history, so one-period
departures suffice.

(ii) By playing $D$ in every period player $i$ secures at least $P_i$
after any history, so every equilibrium continuation gives $i$ at
least $P_i$. Consider a perfect public equilibrium in which $s^C$ is
played in every period. Because there are no false alarms, the public
histories on its path are those without an alarm. If player $i$
departs once, then with probability $1-\psi_i$ no alarm occurs, play
continues as on the path, and $i$ obtains $R_i$ by following the
arrangement from then on; with probability $\psi_i$ an alarm occurs
and $i$ obtains at least $P_i$. The departure is unprofitable only if
$R_i \ge (1-\delta)(R_i + G_i) + \delta[(1-\psi_i) R_i + \psi_i P_i]$,
which is the inequality of part (i). Conversely, the trigger
arrangement is such an equilibrium when the inequality holds.

(iii) As in part (ii), every equilibrium continuation gives player $i$
at least $P_i$, and the others hold $i$ to $P_i$ by playing $D$,
because $s^D$ is a stage-game equilibrium. Permanent play of $s^C$
therefore requires $R_i \ge (1-\delta)(R_i + G_i) + \delta P_i$;
otherwise a one-period departure followed by any equilibrium
continuation is profitable. Conversely, when the inequality holds for
every player, the arrangement ``cooperate as long as $s^C$ has been
played in every past period, and defect forever otherwise'' is a
subgame-perfect equilibrium, by the computation of part (i) with
$\psi_i = 1$.
\end{proof}

Part (i) tests one arrangement. Part (ii) shows that, for permanent
full cooperation, no other arrangement in public strategies does
better when a defector cannot be held below $P_i$; this holds whenever
defection by others lowers a player's payoff, as Condition 2
requires. Without that assumption, punishments harsher than a return
to $s^D$ may sustain cooperation where the trigger arrangement does
not. Part (iii) is the corresponding statement when actions are
observed. Beyond these two information structures, the proposition
does not cover signals with false alarms or outside sanctions; and no
part covers partial or temporary cooperation or strategies that are
not public, so none shows that only $s^D$ can occur when the condition
fails. Trigger strategies, the one-period departure principle
and perfect public equilibrium are standard \citep[Chapters 2 and
  7]{mailath2006}.

\paragraph{Persistence (hypothesis).} The discount factor $\delta$ and
the detection probabilities $\psi_i$ are not measured in either
application. We hypothesize that when the enforcement obstacles of
Condition 5 are severe, the condition of Proposition
\ref{prop:trigger} fails for at least one actor and, more generally,
the Pareto-dominated outcome persists over time: without an authority,
reliable monitoring or credible commitment, agreements to cooperate
are hard to sustain. Whether cooperation is sustained in a given case
depends on its institutions and on what the actors can observe, which
lie outside the proposition.

\begin{remark}[Correction]
Earlier versions of this paper stated Theorem \ref{thm:main} as an
equivalence for all games and included persistence as a fourth
conclusion. The equivalence holds for symmetric anonymous games (part
(d)) but need not hold when players are unequal (part (c)). Condition
5 enters only through
Proposition \ref{prop:trigger} and the hypothesis that follows it.
\end{remark}

Theorem \ref{thm:main} does not by itself show that the outcome it
identifies is resilient against interventions, or that a proposed
solution which leaves the conditions in place must fail: those are
questions about repeated interaction and institutions, which we treat
as hypotheses in the applications below. The diagnostic use of the
framework is therefore heuristic. It directs attention to the
incentive (Condition 3), the welfare comparison (Condition 4) and the
enforcement obstacles (Condition 5) as possible targets that an
intervention might address.

\subsection{Assessment of Canonical Cases}
\label{sec:canonical}

We now assess the five conditions across established coordination
failures in multiple domains, using the cited accounts. For each case
Table \ref{tab:validation} gives the actors and their two actions, the
externality, the private incentive to defect, the reason every actor
would prefer universal cooperation, and the enforcement obstacle. The
entries are the author's qualitative assessments of the cited
accounts, not verified instances of a formal model. Each row describes
the situation absent the remedy named in its last column.

{\setlength{\tabcolsep}{4pt}
\setlength{\LTcapwidth}{\textwidth}
\begin{longtable}{@{} Q{2.1cm} Q{2.3cm} Q{2.2cm} Q{3.3cm} Q{2.4cm} Q{2.4cm} @{}}
\caption{The Five Conditions C1--C5 Across Canonical Coordination
  Failures, including the two new ones introduced in this
  paper. Entries are the author's qualitative assessment of the cited
  accounts; each row describes the situation absent the remedy named
  in its last column.}
\label{tab:validation} \\
\toprule
\textbf{Case} & \textbf{Actors; cooperate / defect (C1, C2)} &
\textbf{Externality of defecting (C2)} & \textbf{Private incentive to
  defect (C3)} & \textbf{Why all prefer universal cooperation (C4)} &
\textbf{Enforcement obstacle; remedy where one exists (C5)} \\
\midrule
\endfirsthead
\multicolumn{6}{@{}l}{\footnotesize\textit{Table \thetable, continued}} \\
\toprule
\textbf{Case} & \textbf{Actors; cooperate / defect (C1, C2)} &
\textbf{Externality of defecting (C2)} & \textbf{Private incentive to
  defect (C3)} & \textbf{Why all prefer universal cooperation (C4)} &
\textbf{Enforcement obstacle; remedy where one exists (C5)} \\
\midrule
\endhead
\bottomrule
\endfoot
\bottomrule
\multicolumn{6}{@{}Q{16.0cm}@{}}{\rule{0pt}{2.8ex}$^{\dagger}$Defection
  is a best response when others are expected to defect, but is not
  strictly dominant in every cited account: the model of
  \citet{diamond1983} also has an equilibrium without a run, and the
  security dilemma of \citet{jervis1978} can take the form of a Stag
  Hunt. For the two cases introduced in this paper, see the
  qualifications in Sections \ref{sec:productivity} and
  \ref{sec:ai}.} \\
\endlastfoot
Commons \citep{hardin1968} & Shepherds on a common pasture; hold the
herd / add sheep & Overgrazing lowers every herd's yield & Each gains
from an added animal whatever the others do (individual profit) &
Restraint preserves the pasture for all & No property rights or
enforceable agreement \\
\addlinespace
Security$^{\dagger}$ \citep{jervis1978, mearsheimer2001} & States;
stay unarmed / arm & One state's arms threaten the security of the
others & Prisoner's-dilemma reading: arming is better whatever the
others do, for survival under anarchy. Stag Hunt reading: only when
the others are expected to arm & Universal disarmament gives security
without the cost of arms & International anarchy: no enforceable
disarmament agreement \\
\addlinespace
Bank runs$^{\dagger}$ \citep{diamond1983} & Depositors; keep deposits
/ withdraw & Withdrawals drain the bank and spread panic & Withdrawing
is a best reply when others are expected to withdraw (fear of losing
deposits); an equilibrium without a run also exists & If no one runs,
every deposit is safe & No coordination during a panic. Deposit
insurance removes the run equilibrium, up to the insured limit \\
\addlinespace
Space debris \citep{adilov2015} & States and firms; limit launches /
expand operations & Debris from launches threatens all satellites &
Each gains commercial or security advantages from launching, whatever
the others do & Restraint preserves usable orbits for all &
International space law lacks enforcement mechanisms \\
\addlinespace
Antibiotics \citep{foster2006} & Doctors, patients, farmers; limit use
/ use freely & Overuse breeds resistance that affects all & Each gains
an immediate health or productivity benefit from use, whatever the
others do & Restraint preserves the drugs' effectiveness for all & Use
is distributed across many sectors \\
\addlinespace
Public goods \citep{olson1965} & Individuals; contribute / withhold &
A contribution withheld lowers provision for all, while the free-rider
still enjoys what others provide & Contributing is costly and changes
the total little, so each gains by free-riding & Universal
contribution funds the good and leaves everyone better off & No
coercive mechanism; taxation is the standard remedy \\
\addlinespace
Hockey helmets \citep{schelling1973} & Players; wear a helmet / play
without & Playing without raises the risk of injury for all & A slight
competitive edge without a helmet (better vision, `toughness'),
whatever the others do & Universal use cancels the edge and makes
everyone safer & No binding rule; a league-wide mandate is the
remedy \\
\addlinespace
Climate change \citep{ostrom2010} & Nations and firms; cut emissions /
pollute freely & Global climate damage & Polluting keeps economic
competitiveness, whatever the others do & Restraint avoids
catastrophic climate costs for all & International anarchy and
sovereignty: no binding, enforceable agreement \\
\midrule
Productivity (this paper) & Firms; cut hours at unchanged output /
keep hours and expand output & Expansion takes market share from
rivals & Universal expansion is the unique equilibrium if and only if
\eqref{eq:Ei} is strict for every firm; dominance has a separate,
share-dependent test (Proposition \ref{prop:dominance}) & Firms save
operating cost and workers gain leisure, for the modeled firms and
workers (Proposition \ref{prop:pareto-productivity}) & No authority
across jurisdictions; hours hard to verify (Section
\ref{sec:productivity}) \\
\addlinespace
AI governance (this paper) & Developers and states; slow development
for safety / accelerate & Risk for all and pressure on rivals &
Argued, not derived; the problem may be one of assurance (Section
\ref{sec:ai}) & Argued under stated assumptions about risk (Section
\ref{sec:ai}) & Software hard to verify; no global authority \\
\end{longtable}
}

In our assessment, each canonical case exhibits the five conditions,
with two qualifications: in bank runs and in the Stag Hunt form of the
security dilemma, defection is a best response to expected defection,
not a dominant strategy. This suggests the framework captures a
structure common to many coordination failures.

Several of these cases were later addressed by an institution that
changed the incentive to defect. Deposit insurance (FDIC) is the
clearest example: under the assumptions of Diamond and Dybvig's
analysis it removes the run equilibrium, for deposits up to the
insurance limit, because insured depositors no longer gain from
withdrawing when others withdraw. Wholesale funding and equities
markets are not covered, so they may still be vulnerable to runs. Such
cases are consistent with the diagnostic use of the framework; they do
not by themselves validate it.

\subsection{Extension: Condition Intensity and Tragedy Severity}

While the framework treats conditions as binary (present or absent)
for analytical clarity, each condition varies in intensity in
real-world applications. We grade the intensity of each condition on
a common four-level scale, set out in Table
\ref{tab:intensity-scale}. The level
descriptions are the author's grading rubric, and the statements about
how intensity affects severity are the hypotheses that this rubric
encodes, not established results. The intensity dimension is intended
to distinguish milder cases from more severe ones. The third dimension
is not Condition 3 itself. We write C3* for it: it rates the strength
of the incentive to defect when others defect, and it does not require
the dominance of Condition 3.

{\setlength{\tabcolsep}{4pt}
\setlength{\LTcapwidth}{\textwidth}
\begin{longtable}{@{} Q{1.9cm} Q{3.25cm} Q{3.25cm} Q{3.25cm} Q{3.25cm} @{}}
\caption{Intensity Levels of Conditions C1--C5: the Author's Grading
  Rubric on a Four-Level Scale. The third dimension, C3*, rates the
  strength of the incentive to defect when others defect; it does not
  require the dominance of Condition 3.}
\label{tab:intensity-scale} \\
\toprule
\textbf{Dimension} & \textbf{Low} & \textbf{Medium} & \textbf{High} & \textbf{Extreme} \\
\midrule
\endfirsthead
\multicolumn{5}{@{}l}{\footnotesize\textit{Table \thetable, continued}} \\
\toprule
\textbf{Dimension} & \textbf{Low} & \textbf{Medium} & \textbf{High} & \textbf{Extreme} \\
\midrule
\endhead
\bottomrule
\endfoot
\bottomrule
\endlastfoot
C1: number and heterogeneity of actors & Few actors with stable
identities (e.g., small groups like hockey leagues, bilateral security
relationships) & Moderate number of organized actors (e.g., dozens of
states or firms in defined sectors like commons governance, public
goods provision) & Large number of heterogeneous actors across sectors
(e.g., hundreds of nations plus countless firms in climate change,
global productivity competition) & Massive number of individual-level
decisions across decentralized populations (e.g., millions of
depositors in bank runs, billions of patients and prescribers in
antibiotics resistance) \\
\addlinespace
C2: magnitude and scope of externalities & Small, localized negative
externalities with delayed or uncertain effects (e.g., individual
free-riding on public goods, minor competitive disadvantages in hockey
helmets) & Moderate externalities with gradual accumulation (e.g.,
commons depletion over years, antibiotic resistance developing
gradually) & Large, immediate externalities with rapid transmission
(e.g., financial panic spreading within days, climate damages
affecting all actors, AI capability advantages immediately threatening
competitors) & Anticipatory externalities where actors respond to
perceived threats before they materialize, creating self-reinforcing
spirals (e.g., security dilemmas where intelligence about potential
programs triggers immediate counter-programs) \\
\addlinespace
C3*: strength of the incentive to defect (stakes) & Defection provides modest
advantage; cooperation causes minor competitive disadvantage (e.g.,
visibility or comfort in hockey helmets, reduced returns in public
goods) & Defection provides significant competitive advantage;
cooperation causes substantial market share loss but not immediate
elimination (e.g., productivity competition where firms can survive
with reduced hours, commons where individuals can survive with
restraint, climate where nations face competitive disadvantage but can
adapt) & Defection necessary to avoid severe consequences; cooperation
risks major losses (e.g., state conquest in security dilemmas, losing
all deposits in bank runs) & Defection necessary to avoid permanent
strategic obsolescence; restraint means irreversible loss with no
possibility of recovery (e.g., transformative AI where falling behind
creates permanent winner-take-all dynamics) \\
\addlinespace
C4: magnitude of collective loss & Universal cooperation slightly
better than universal defection; modest welfare gains available (e.g.,
marginal safety improvements from universal helmet use) & Universal
cooperation moderately better; significant quality-of-life
improvements but society continues functioning (e.g., commons
preservation, public goods provision, productivity welfare gains,
avoiding bank runs, maintaining antibiotic effectiveness) & Universal
cooperation vastly superior; defection leads to major catastrophic
outcomes with potential for recovery (e.g., regional conflicts,
economic depressions) & Universal cooperation necessary to avoid
irreversible catastrophic outcomes with permanent losses (e.g., human
extinction from nuclear war or AI misalignment, irreversible climate
tipping points causing ecosystem collapse) \\
\addlinespace
C5: severity of structural barriers to enforcement & Enforcement
absent or weak in current form but structurally feasible; clear
authority could monitor and sanction violations with institutional
development & Enforcement difficult but possible with institutional
investment; monitoring costly but achievable; commitment requires
sustained effort (e.g., league rules for helmets, taxation for public
goods, local commons that a community can monitor) & Enforcement very difficult due to international anarchy or
distributed nature; verification problematic; commitment fragile;
sanctions difficult but not impossible (e.g., international labor standards, nuclear non-proliferation,
climate agreements, bank run coordination) & Enforcement judged
infeasible with existing means, because of severe technical barriers
(concealment, instant copying, inseparable dual-use) combined with
international anarchy and no feasible global authority (e.g., in our
assessment, AI development, where software is copyable and training
and capabilities are hard to observe and verify) \\
\end{longtable}
}

Enforcement difficulty has three components, each running from mild to
severe: authority (strong, weak, none), verification (easy monitoring,
difficult detection, perfect hiding) and commitment (binding, costly,
not credible).

The rubric encodes the following hypotheses about severity. More
actors can complicate coordination by making agreements harder to
organize and monitor, and heterogeneity among them (different sizes,
capabilities, preferences) adds to the difficulty; how group size
changes individual incentives depends on the payoff structure. Larger
externalities create stronger feedback loops: when one actor's
defection severely harms others, it triggers more aggressive
competitive responses, and the speed and magnitude of transmission
determine how rapidly the tragedy unfolds. The intensity of
existential stakes determines how strong the incentive to defect is:
when cooperation means death or bankruptcy, actors will defect
regardless of collective costs, while lower stakes create space for
partial coordination or voluntary restraint. Greater
Pareto-inefficiency means larger collective losses from coordination
failure, which affects both the normative urgency and the practical
difficulty (actors may invest more in solving severe tragedies); a
larger surplus from cooperation can also make cooperation easier to
sustain (Proposition \ref{prop:trigger}), and the grade measures what
is lost when it fails. The enforcement grade is meant to separate
cases where enforcement is difficult but possible, where we expect
sufficient institutional investment to be able to achieve
coordination, from cases where it appears infeasible with existing
means, where we expect coordination to require a change in the
underlying structure, such as a new verification technology or a new
authority.

\subsection{Uses and Limits of the Framework}
\label{sec:uses}

The checklist and the intensity grades suggest several analytical
uses. We offer them as heuristics, not as deductions from Theorem
\ref{thm:main}.

\paragraph{Identification:} We can identify candidate tragedies by
checking whether a situation satisfies the framework. If it does, we
expect coordination to be difficult. If a condition fails, the
situation falls outside the framework; this does not by itself imply
that coordination is feasible.

\paragraph{Comparison and Prediction:} We can compare domains by their
condition intensities, as a heuristic for the relative difficulty of
coordination. We hypothesize that situations with extreme stakes or
severe enforcement problems resist coordination more strongly than
situations where conditions are weakly satisfied. AI governance (for
transformative breakthroughs) is rated at extreme intensity on three
dimensions simultaneously (C3*, C4, C5), while severe historical cases
like nuclear security dilemmas are rated at extreme intensity on two
conditions (C2, C4) and climate change on one (C4). Condition
intensities also change over time. Productivity
competition shows moderate base intensity (C2, C3* at medium due to
gradual market dynamics and non-existential stakes), but globalization
after 1980 increased $n$ (more actors), strengthened externalities
(faster competitive transmission), and weakened enforcement (capital
mobility), making even this moderate tragedy harder to
address. Climate change intensified as C4 shifted from high to extreme
with growing evidence of irreversible tipping points.

\paragraph{Intervention Design:} Interventions can target specific
conditions, in the first place those of highest intensity. Changing a
condition, for example removing the dominance of defection, is a
natural target, although it does not guarantee that coordination
follows. For AI, the extreme intensity of C5 (enforcement obstacles
due to verification) suggests compute governance (making verification
partially possible) may be most promising despite remaining severe
challenges.

\paragraph{Falsification:} We hypothesize that, other relevant features
held fixed, lowering condition intensity tends to reduce coordination
difficulty. Systematic evidence that such reductions do not improve
coordination outcomes would count against
this hypothesis. The European case is consistent with it: lower
intensity on some conditions (reduced existential stakes through
wealthy welfare states, supranational governance reducing anarchy)
accompanies partial coordination, despite high intensity on others
(global competition, verification difficulty).

\subsection{Tragedy Index: Intensity Quantification}
\label{sec:tragedy-index}

The Tragedy Index ($\iota$) summarizes the assessed condition
intensities in one number. It grades five dimensions, one for each
condition, and is positive only when all five are rated present; it
summarizes the checklist and is not a second definition of a
structural tragedy. Figure \ref{fig:radar_chart} plots the profiles of
four cases.

\begin{figure}[!htb]
\centering
\includegraphics[width=0.42\textwidth]{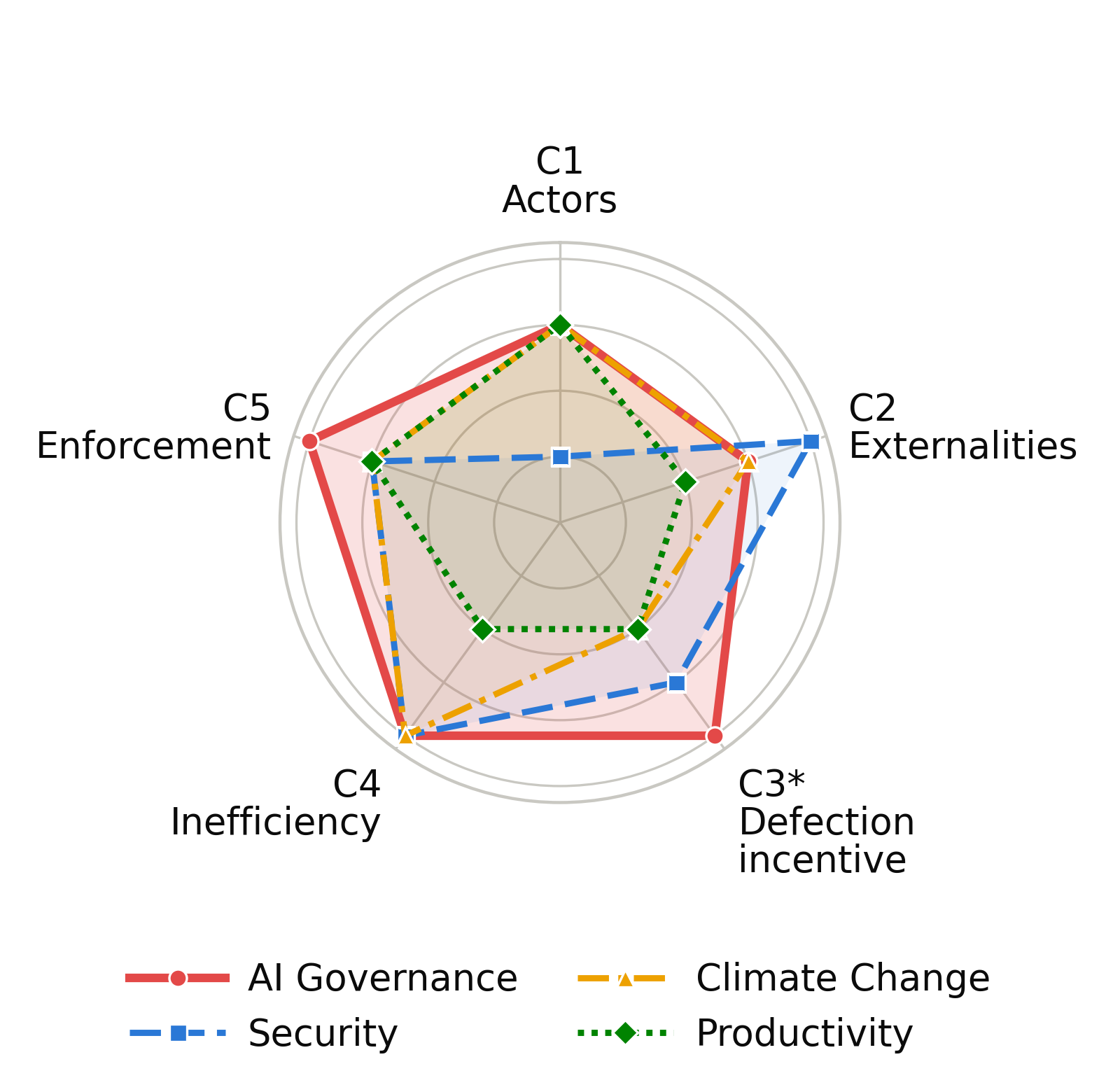}
\caption{\textbf{The Structural Fingerprint of AI Governance.}
  Assessed condition intensities for four selected domains
  (logarithmic scale), chosen to illustrate diverse structural
  profiles. The rings mark the four grades, from Low at the innermost
  ring to Extreme at the outermost. AI Governance (red) is rated at
  or above the other cases on every dimension except C2, where
  Security is higher, whereas Security and Climate Change spike only
  on specific dimensions. Overlapping areas (blended colors) indicate
  shared structural challenges.}
\label{fig:radar_chart}
\end{figure}

\paragraph{Status of the Index:} The Tragedy Index is a heuristic
framework for comparing the intensity of the five conditions across
domains. The ratings are the author's assessments against the level
descriptions in Table \ref{tab:intensity-scale}, and the scale, the
equal weights and the
aggregation rule are modeling choices. The Index orders domains by
assessed intensity; its scores are not measured probabilities or
durations of coordination failure. One dimension needs care in
interpretation: a larger cooperative surplus (Condition 4) raises the
Index, although it can also make cooperation easier to sustain, by
raising what an actor loses when cooperation breaks down (Proposition
\ref{prop:trigger}). The third dimension, C3*, is broader than
Condition 3: bank runs, the Stag Hunt form of the security dilemma and
the larger firm in Example \ref{ex:two-firm} have an incentive to
defect when others defect without dominance. A positive Index
therefore establishes neither dominance nor uniqueness nor
persistence.

\paragraph{Scale and Aggregation:} We encode the grades as powers of
ten, $I_c \in \{0, 10, 10^2, 10^3, 10^4\}$ for Absent, Low, Medium,
High and Extreme, so that the steps between them are multiplicative:
we regard a condition rated `Extreme' as far more intense than one
rated `Low', not merely a few times more. The base-10 choice is
somewhat arbitrary. We combine the five intensities by their geometric
mean:
\begin{equation}\label{eq:index}
\iota = \left( \prod_{c=1}^{5} I_c \right)^{\frac{1}{5}}
\end{equation}
We take the fifth root so that the Index is on the scale of a single
condition's intensity: five conditions all rated Medium give 100, the
intensity of Medium. The index equals ten raised to the average grade
(counting Low as 1, Medium as 2, High as 3 and Extreme as 4), so
raising any one condition by one grade multiplies the index by
$10^{1/5} \approx 1.58$, and profiles with the same average grade
receive the same index. Unlike an arithmetic mean of the intensities,
the geometric mean is not dominated by a single extreme rating: one
Extreme condition with four Low ones gives an index of about 40, below
the 100 of five Medium conditions, whereas the arithmetic mean of the
same intensities would be about 2,000. This is our reason for the
choice: a lone extreme condition should not outweigh moderate
intensity across all five. If any dimension has intensity zero
(Absent), the index is zero under this convention. For positive
intensities, alternative aggregations such as a logarithmic sum would
yield different absolute values but preserve ordinal rankings across
domains.

We label ranges of the index as ``Tragedy Magnitude'':
\begin{itemize}
\item Magnitude 0: a dimension is rated absent
\item Magnitude 10 to under 100: Manageable Nuisance (Chronic)
\item Magnitude 100 to under 1,000: Systemic Crisis (Acute)
\item Magnitude 1,000 and above: Very high composite intensity
\end{itemize}

As shown in Table \ref{tab:intensity-comparison}, AI Governance
receives an index of $\iota = 3{,}981$, the highest among the cases
assessed and about four times that of the next-highest cases under the
base-10 encoding (a difference of 0.6 in the average grade). Climate
change, bank runs, and nuclear security dilemmas receive the same
composite score ($\iota = 1{,}000$ each) with different profiles:
security through anticipatory dynamics and extinction-level risks
despite few actors; climate through irreversibility; bank runs through
mass-scale coordination failure. Productivity competition and
antibiotics resistance, despite their global scale and persistence,
score lower because of non-existential stakes and gradual
dynamics. Low-scoring cases like public goods have often been
addressed through institutional solutions.

\begin{table}[t]
\centering
\caption{Intensity Comparison of Conditions C1--C5 Across Domains,
  including the two new ones introduced in this paper. Intensities:
  L=Low (10), M=Medium (100), H=High (1,000), E=Extreme
  (10,000).\textsuperscript{a}}
\label{tab:intensity-comparison}
\begin{tabular}{lcccccc}
\toprule
\textbf{Domain} & \textbf{C1} & \textbf{C2} & \textbf{C3*} & \textbf{C4} & \textbf{C5} & \textbf{Index ($\iota$)} \\
\midrule
Commons & M & M & M & M & M & \textbf{100} \\
Security\textsuperscript{b} & L & E & H & E & H & \textbf{1,000} \\
Bank Runs & E & H & H & M & H & \textbf{1,000} \\
Space Debris & M & M & M & M & H & \textbf{158} \\
Antibiotics & E & M & M & M & H & \textbf{398} \\
Public Goods & M & L & L & M & M & \textbf{40} \\
Hockey Helmets & L & L & L & L & M & \textbf{16} \\
Climate Change & H & H & M & E & H & \textbf{1,000} \\
\midrule
Productivity & H & M & M & M & H & \textbf{251} \\
\textbf{AI Governance}\textsuperscript{c} & \textbf{H} & \textbf{H} & \textbf{E} &
\textbf{E} & \textbf{E} & \textbf{3,981} \\
\bottomrule
\end{tabular}

\vspace{0.5em} \footnotesize \textsuperscript{a}Ratings are the
author's assessments against the level descriptions in Table
\ref{tab:intensity-scale}. C3*
rates the strength of the incentive to defect when others defect; it
does not require dominance (Security in its Stag Hunt form, Bank
Runs).

\textsuperscript{b}Security dilemmas rated for nuclear-age scenarios
where C2 reflects anticipatory responses to perceived threats and C4
reflects extinction-level risks.

\textsuperscript{c}AI Governance ratings reflect transformative
breakthroughs (AGI/superintelligence). As an illustrative profile for
current AI capabilities competition, (H, H, H, H, E) gives an Index of
about 1,585, somewhat above climate change and bank runs; it assumes
major but recoverable losses, as multiple competitors survive
capability gaps.
\end{table}

\subsection{Alternative Formalization: Game-Theoretic Measures}
\label{sec:game-theoretic-measures}

While the Tragedy Index provides intuitive comparison across domains
through structural condition intensities, formal game-theoretic
measures offer an alternative approach to quantifying tragedy
severity. These measures focus on the payoff structure (Conditions 3
and 4) and provide mathematical precision about welfare losses and
incentive misalignments.

Consider an N-player game satisfying Conditions 1--5. Let $W(s)$ denote
global welfare, the sum of the players' payoffs, at profile $s$
(the expected sum at a mixed profile). Let
$U_{coop} = W(s^C)$ denote global welfare under universal cooperation,
$U_{nash} = W(s^D)$ denote global welfare at the Nash equilibrium
(profile $s^D$), and for symmetric games, let $T$ be the temptation payoff
(defect while others cooperate), $R$ be the reward payoff (mutual
cooperation), $P$ be the punishment payoff (mutual defection), and $S$
be the sucker payoff (cooperate while others defect).  In a
symmetric two-player game, Conditions 3 and 4 together are the
prisoner's dilemma ordering $T > R > P > S$.
The ratios below require payoffs in meaningful, comparable welfare
units with a meaningful origin, and the stated ranges and labels hold
only under the sign conditions given with each measure.

\paragraph{The Price of Anarchy:} This measure
\citep{koutsoupias1999} is the ratio between optimal social welfare and worst-case equilibrium
welfare:
\begin{equation}\label{eq:poa}
\text{PoA} = \frac{\max_{s} W(s)}{\min_{s \in \mathrm{NE}} W(s)}
\end{equation}
where NE is the set of Nash equilibria and equilibrium welfare is
positive. When $s^D$ is the only equilibrium and universal cooperation
maximizes total welfare, this is $U_{coop}/U_{nash}$, which equals
$R/P$ in the symmetric case with additive welfare. Condition 4 alone
does not make universal cooperation the welfare maximum: in the
symmetric two-player game with $T = 10$, $R = 3$, $P = 1$ and $S = 0$,
total welfare is 6 under mutual cooperation and 10 when exactly one
player defects, so $U_{coop}/U_{nash} = 3$ while the Price of Anarchy
is 5. Equivalently, $1 - 1/\text{PoA}$ is the fraction of the optimal
welfare that is lost at the worst equilibrium. As illustrative labels,
PoA $\approx$ 1 indicates near-efficiency, PoA $\in [10, 100]$ severe
inefficiency, and, for a fixed positive optimal welfare, PoA $\to
\infty$ an equilibrium whose value approaches zero.

\paragraph{The Temptation-Cooperation Ratio:} This ratio captures the tension between
individual incentives and collective benefits (in the spirit of the
payoff-ratio indices of \citet{rapoport1965}):
\begin{equation}\label{eq:tau}
\tau = \frac{T - R}{R - P}
\end{equation}
It compares the individual gain from defection
with the collective gain from cooperation. As illustrative labels, we
read $\tau < 1$ as a weaker temptation relative to the gain from
cooperation, $\tau \approx 1$--5 as a strong temptation, and $\tau \to
\infty$ as a temptation that is very large relative to the cooperative
surplus (which can also arise because $R - P$ is small while the
stakes are bounded). For identical players $\tau = G/v$ in the
notation of Section \ref{sec:formal-result}, and it is the quantity
that Proposition \ref{prop:trigger} compares with deterrence: with a
common detection probability $\psi$, the trigger arrangement sustains
full cooperation exactly when $\tau \le \psi\,\delta/(1-\delta)$. A
value of $\tau$ below one is therefore not
by itself a criterion that coordination is feasible.

These measures require payoff specifications for each domain. Section
\ref{sec:productivity} supplies them for productivity competition; for
AI governance, security dilemmas, and climate change they would
require modeling uncertainty about catastrophic outcomes, discount
rates, and long-term value functions, which is beyond our scope. The
Tragedy Index and these measures are complementary: the Index compares
diverse domains without payoff quantification, and the measures give a
welfare analysis where payoffs can be reliably estimated.

\section{The Tragedy of Productivity}
\label{sec:productivity}

We now apply the unified framework to productivity competition,
arguing heuristically that the five conditions are present, and giving
in a formal model the exact condition under which universal expansion
of output is the unique, Pareto-dominated equilibrium. This offers a structural explanation of why Keynes's
expectation has not been realized: coordination barriers in
competitive markets, in addition to cultural preferences or policy
failures.

Although worker welfare includes wages and benefits in addition to
working hours, we will use working hours as a proxy for worker welfare
as well as to simplify the formal analysis. Our welfare analysis
focuses on firms and workers as the direct participants in labor
market coordination. Productivity gains may generate broader societal
effects (e.g., consumer surplus from lower prices), but these lie
outside our scope: buyers are not modeled. Our results therefore
concern the modeled firms and their workers and do not establish a
loss of total welfare.

\subsection{Mapping the Five Conditions to Firm Competition}
\label{sec:productivity-mapping}

\subsubsection{Condition 1: N-Player Structure}

Multiple firms compete in domestic and global markets, national
governments face political pressure from firms and workers, and no
single actor can coordinate welfare improvements unilaterally. The
decision architecture is decentralized.

\subsubsection{Condition 2: Binary Choice with Externalities}

When productivity increases, each firm faces essentially a binary choice:
\begin{itemize}
    \item Strategy $C$ (Cooperate/Restrain): Reduce working hours
      proportionally to productivity gains while maintaining output
    \item Strategy $D$ (Defect/Expand): Maintain or increase working
      hours to expand output and market share
\end{itemize}

The externality is negative: a firm that expands output captures
market share from its competitors, reduces their profit margins, and
puts them under pressure to expand as well.

\subsubsection{Condition 3: Dominance Property}

We argue that defection (output expansion) dominates cooperation
(welfare improvements) regardless of competitors' choices. The
argument here is heuristic; the formal model in the next subsection
states the exact condition under which universal expansion is the only
equilibrium, and the stronger condition under which expansion is
dominant:

\textit{Case 1—If competitors reduce hours:} The firm that expands
gains market share relative to rivals who keep their baseline output,
increasing its profits and market position.

\textit{Case 2—If competitors expand output:} The firm must expand to
avoid displacement. Competitors' expansion increases supply and
reduces prices. Reducing hours while competitors expand means holding
output at its baseline in a more competitive market—a recipe for
bankruptcy.

This dominance is often driven by existential stakes: in market
competition, by the threat of insolvency. A firm that maintains
high-cost working conditions while competitors expand output and lower
prices risks losing market share to the point of bankruptcy. As
\citet{alchian1950uncertainty} famously argued, market competition
functions as an evolutionary selection mechanism; firms that fail to
maximize profit risk elimination.

\subsubsection{Condition 4: Pareto-Inefficiency}

For the modeled firms and their workers, universal welfare
improvements Pareto-dominate universal output expansion (Proposition
\ref{prop:pareto-productivity}). If all firms reduce hours
proportionally to productivity gains, workers gain leisure time
without income loss and firms maintain their relative market
positions; under mutual expansion workers work longer hours despite
productivity gains and firms engage in costly competition. Wider
effects of the two outcomes, such as the sufficiency of total output,
excess supply, and the environmental effects of production, lie
outside the model. The equilibrium is recognizably suboptimal: the
capacity for leisure exists but remains unrealized.

\subsubsection{Condition 5: Enforcement Difficulty}

Four structural barriers stand in the way of coordination.
\textit{International anarchy:} no global authority can enforce
coordinated work reduction, and capital mobility lets firms relocate
to jurisdictions with longer hours or fewer restrictions, undermining
any single nation's ability to impose work reductions that competitors
do not match. \textit{Verification challenges:} firms can disguise
overtime, create parallel structures, or use contractors, and workers
facing job insecurity may voluntarily work longer hours, so actual
working hours may diverge substantially from official
regulations. \textit{Monitoring infeasibility:} remote work, knowledge
work, and always-on digital connectivity put work outside traditional
workplace boundaries, and enforcement costs escalate.
\textit{Commitment problems:} each firm has a strong incentive to
defect once others reduce hours, and the firm that reduces hours first
suffers competitive disadvantage if others do not follow, so without
binding enforcement any agreement is unstable.

\subsubsection{Result}

On the heuristic argument above, productivity competition exhibits the
five conditions. The formal model below does not rest on that argument
for Condition 3; it tests the outcome directly, and Proposition
\ref{prop:dominance} gives the exact condition under which universal
expansion is the unique equilibrium.

\subsection{Formalization: The N-Firm Game}

We now formalize the heuristic arguments from the previous section in
a stylized model, and state the exact condition under which the
tragedy structure holds in it.

Consider $n \geq 2$ firms indexed by $i \in \{1, \ldots, n\}$ competing
in a market. Firms are heterogeneous: firm $i$ has initial output
$Q_{i} > 0$ and working hours $H_{i} > 0$ such that productivity is
output per hour. We assume that working hours do not exceed the physiological
threshold where marginal productivity turns negative.

We define two separate utility functions: one for the firm ($U_F$),
driven by profit, and one for the workers ($U_W$), driven by welfare.

\begin{itemize}
    \item \textbf{Firm Utility ($U_F$):} This is the firm's profit,
      up to terms that are the same under both strategies. We model a
      reduction in working hours at unchanged total pay, as in
      salaried work. Each firm's wage bill is therefore the same under
      both strategies; it cancels from the comparison of the two,
      together with other fixed costs, and is omitted. What remains is
      revenue minus the costs that vary with hours. We also ignore
      other typical costs such as taxes and costs that may arise under
      mutual defection such as additional advertising costs and
      price wars.
    \begin{itemize}
        \item $r > 0$ is the total market-wide revenue, assumed
          fixed. Each firm receives a share of $r$ equal to its share
          of total output. Thus $r$ is an amount of money, not a price:
          the price of a unit of output is $r$ divided by total output,
          and it falls as total output rises.
        \item $c > 0$ is the non-wage operating cost per hour worked:
          costs that scale with operating hours, such as energy and
          equipment running costs. Working hours stand in for
          operating hours. Costs proportional to output, such as
          materials, are omitted. With a constant unit cost $m \ge 0$,
          the equilibrium and dominance conditions of Proposition
          \ref{prop:dominance} and the firm's surplus from universal
          restraint retain their form with $c H_i$ replaced by $c H_i
          + \alpha\, m\, Q_i$. Wages are not part of $c$, because total
          pay does not vary with hours in this model.
    \end{itemize}
    \item \textbf{Worker Utility ($U_W$):} This is welfare: $\text{Wage} + \text{Leisure Value}$.
    \begin{itemize}
        \item $\text{Wage}$ is the worker's total pay, which is the
          same in both scenarios.
        \item $\beta > 0$ is the monetized utility value of one hour of
          leisure.
    \end{itemize}
\end{itemize}

Both utility equations above ignore other components that stay
constant between strategy changes.

When productivity increases by factor $\alpha > 1$, each firm chooses
between two strategies, $s_i \in \{C, D\}$:
\begin{itemize}
    \item $C$ (Cooperate): Reduce hours to $H_{i}/\alpha$, maintaining output $Q_{i}$
    \item $D$ (Defect): Maintain hours $H_{i}$, expanding output to $\alpha Q_{i}$
\end{itemize}

\paragraph{Payoff Structure:}

The utility for firm $i$ depends on its own strategy $s_i \in \{C,
D\}$ and the strategy profile of its competitors,
$s_{-i}$. In this specific market game, the impact of
$s_{-i}$ is fully captured by the total output of all other
firms, which we denote as $Q_{-i}(s_{-i})$.

The utility functions for firm $i$ and a worker in the firm are therefore:
\begin{itemize}
    \item \textbf{If Cooperating ($s_i = C$):} Output is $Q_{i}$, Hours are $H_{i}/\alpha$.
        \begin{itemize}
            \item $U_{F,i}(C, s_{-i}) = r \cdot \frac{Q_{i}}{Q_{i} + Q_{-i}(s_{-i})} - c \cdot \frac{H_{i}}{\alpha}$
            \item $U_{W,i}(C) = \text{Wage} + \beta \cdot \frac{(\alpha-1)H_{i}}{\alpha}$
        \end{itemize}
    \item \textbf{If Defecting ($s_i = D$):} Output is $\alpha Q_{i}$, Hours are $H_{i}$.
        \begin{itemize}
            \item $U_{F,i}(D, s_{-i}) = r \cdot \frac{\alpha Q_{i}}{\alpha Q_{i} + Q_{-i}(s_{-i})} - c \cdot H_{i}$
            \item $U_{W,i}(D) = \text{Wage}$
        \end{itemize}
\end{itemize}

For firm $i$ let $Q^{C}_{-i} = \sum_{j \ne i} Q_j$ be its rivals' output
when they all restrain, and let $\sigma_i = Q_i / (Q_i + Q^{C}_{-i})$
be its baseline market share. Rivals' output $x = Q_{-i}(s_{-i})$ lies
between $Q^{C}_{-i}$ (all restrain) and $\alpha Q^{C}_{-i}$ (all
expand), and the firm's gain from expanding, the market-share gain
minus the added operating cost, is
\begin{equation}\label{eq:gain-model}
  g_i(x) = U_{F,i}(D, s_{-i}) - U_{F,i}(C, s_{-i})
  = (\alpha - 1)\left[\frac{r\, Q_i\, x}{(\alpha Q_i + x)(Q_i + x)}
  - \frac{c\, H_i}{\alpha}\right].
\end{equation}

\begin{proposition}[Equilibrium, Uniqueness and Dominance in
    Productivity Competition]\label{prop:dominance}
Let $\alpha > 1$.
\begin{enumerate}
    \item[(i)] Against every competitor profile, defection yields
      strictly higher market share than cooperation.
    \item[(ii)] Universal expansion $s^D$ is a Nash equilibrium if and
      only if, for every firm $i$,
      \begin{equation}\tag{E}\label{eq:Ei}
        c\, H_i \;\le\; \frac{r\,\alpha\,\sigma_i (1 -
          \sigma_i)}{\sigma_i + \alpha (1 - \sigma_i)} .
      \end{equation}
    \item[(iii)] $s^D$ is the unique Nash equilibrium, in pure or
      mixed strategies, if and only if \eqref{eq:Ei} holds with strict
      inequality for every firm.
    \item[(iv)] Expansion is strictly dominant for firm $i$ if and
      only if $g_i > 0$ at the two extreme profiles, all rivals
      restraining and all rivals expanding. For a firm with $\sigma_i
      \le 1/2$ this is the strict form of \eqref{eq:Ei}. For a firm
      with $\sigma_i > 1/2$ it is
      \begin{equation}\label{eq:dominance-large}
        c\, H_i \;<\; \frac{r\,\alpha\,\sigma_i (1 -
          \sigma_i)}{\alpha\,\sigma_i + 1 - \sigma_i} ,
      \end{equation}
      which is stronger than the strict form of \eqref{eq:Ei}.
\end{enumerate}
\end{proposition}

\begin{proof}
(i) Since $\alpha > 1$ and $x > 0$, the share $\alpha Q_{i} / (\alpha
Q_{i} + x)$ is strictly greater than $Q_{i} / (Q_{i} + x)$.

(ii) By Theorem \ref{thm:main}(a), $s^D$ is an equilibrium if and only
if $g_i(\alpha Q^{C}_{-i}) \ge 0$ for every firm. Substituting $x =
\alpha Q^{C}_{-i}$ in \eqref{eq:gain-model} and writing the result
in terms of $\sigma_i$ gives \eqref{eq:Ei}.

(iii) and (iv). Let $f(x) = x / ((\alpha Q_i + x)(Q_i + x))$. Its
derivative has the sign of $\alpha Q_i^2 - x^2$, so $f$ rises and then
falls, and its minimum on the interval $[Q^{C}_{-i}, \alpha
  Q^{C}_{-i}]$ is at an endpoint. Both endpoints are profiles of the
game, which gives the first statement of (iv). Moreover $f(\alpha
Q^{C}_{-i}) \le f(Q^{C}_{-i})$ if and only if $Q^{C}_{-i} \ge Q_i$,
that is, $\sigma_i \le 1/2$. So for a firm with at most half the
market the smallest gain is at ``all rivals expand'', where it is
positive by the strict form of \eqref{eq:Ei}: expansion is strictly
dominant for it. For a firm with more than half the market the
smallest gain is at ``all rivals restrain'', which gives
\eqref{eq:dominance-large}. At most one firm has more than half the
market. In
any equilibrium, in pure or mixed strategies, every other firm expands
with probability one, so that firm faces ``all rivals expand'', where
its gain is positive by the strict form of \eqref{eq:Ei}. Hence $s^D$
is the unique equilibrium when \eqref{eq:Ei} is strict for every firm.

Conversely, suppose $s^D$ is an equilibrium, so that \eqref{eq:Ei}
holds for every firm, with equality for at least one. We construct a
second equilibrium. By the endpoint argument, expansion is weakly
dominant for every firm with at most half the market. Let $i$ be a
firm with equality, taking the firm with more than half the market if
that firm has equality. Let firm $i$ restrain with probability
$\varepsilon > 0$ and expand otherwise, and let all other firms
expand. Firm $i$ faces ``all rivals expand'' and is indifferent, by
equality in \eqref{eq:Ei}. Every other firm with at most half the
market expands by weak dominance. If a firm $j \ne i$ has more than
half the market, \eqref{eq:Ei} is strict for it, so its gain $a$ when
all rivals expand is positive; if $b$ is its gain when $i$ restrains
and the others expand, its expected gain $(1-\varepsilon)a +
\varepsilon b$ is positive for $\varepsilon$ small enough. The
profile is therefore an equilibrium other than $s^D$, and $s^D$ is not
unique.
\end{proof}

Condition \eqref{eq:Ei} says that the revenue firm $i$ would lose by
restraining alone, when all rivals expand, is at least the operating
cost it would save. Universal expansion is an equilibrium exactly when
\eqref{eq:Ei} holds for every firm, and it is the only equilibrium
exactly when \eqref{eq:Ei} holds strictly for every firm. Expansion
need not be dominant for every firm: the single
inequality \eqref{eq:Ei} does not characterize dominance for a firm
with more than half the market.

\begin{example}[Uniqueness without Dominance]\label{ex:two-firm}
Let $n = 2$, $Q = (3, 1)$, $H = (1, 1/2)$, $\alpha = 2$, $r = 100$ and
$c = 25$. Firm 2 prefers to expand against either action of firm
1. Firm 1 prefers to restrain if firm 2 restrains and to expand if
firm 2 expands. Universal expansion is the unique equilibrium, and
both firms would be better off under universal restraint. The
parameters are illustrative.
\end{example}

\begin{corollary}[Equal Output per Hour]\label{cor:equal-productivity}
Suppose all firms have the same output per hour, and let $\kappa = c
\sum_j H_j / r$ be the industry's hourly-cost bill at the old hours as
a share of revenue. Then \eqref{eq:Ei} holds strictly for every firm
if and only if $\kappa < 1$ and the largest firm's market share is
below $\alpha (1 - \kappa) / (\alpha (1 - \kappa) + \kappa)$.
\end{corollary}

\begin{proof}
With equal output per hour, $H_i = \sigma_i \sum_j H_j$, so the strict
form of \eqref{eq:Ei} reads $\kappa\,(\sigma_i + \alpha (1 -
\sigma_i)) < \alpha (1 - \sigma_i)$, that is, $\kappa\,\sigma_i <
\alpha (1 - \kappa)(1 - \sigma_i)$. This fails for every firm if
$\kappa \ge 1$. If $\kappa < 1$ it holds exactly when $\sigma_i$ is
below the stated bound, and it holds for every firm when it holds for
the largest.
\end{proof}

For example, with $\alpha = 1.5$ the bound on the largest share is
about $0.86$ when $\kappa = 0.2$ and about $0.27$ when $\kappa = 0.8$;
the values are illustrative. A lower hourly operating cost makes
\eqref{eq:Ei} easier to satisfy. It also lowers what firms gain from
universal restraint, $c H_i (1 - 1/\alpha)$, which vanishes as $c$
tends to zero; the workers' gain in leisure does not depend on $c$.

\paragraph{The temptation ratio in the model.} Suppose the firms are
identical, with the same output and the same hours $H$, and let
$\kappa = c\,n\,H/r$ as in Corollary \ref{cor:equal-productivity}. The
ratio of Section \ref{sec:game-theoretic-measures} is then
\begin{equation}\label{eq:tau-model}
  \tau = \frac{G}{v} = \frac{\alpha\,(n-1)}{\kappa\,(\alpha + n - 1)} - 1 .
\end{equation}
Holding $\alpha$ and $\kappa$ fixed, it rises with the number of firms
toward $\alpha/\kappa - 1$, and it is large when the hourly operating
cost is small relative to revenue. Where \eqref{eq:Ei} holds strictly,
Proposition \ref{prop:trigger} applies to the firms. Its additional
payoff condition also holds: with its own expansion fixed, a firm's
payoff falls as its rivals' total output rises, so $U_{F,i}(D, s_{-i})
\ge P_i$ for every $s_{-i}$. Thus parts (ii) and (iii) apply under
their respective monitoring assumptions. With a common detection
probability $\psi$, the trigger arrangement sustains universal
restraint exactly when $\tau \le \psi\,\delta/(1-\delta)$, that is,
when $\delta \ge \tau/(\tau + \psi)$. For example, $n = 10$, $\alpha =
1.5$ and $\kappa = 0.2$ give $\tau = 38/7 \approx 5.43$, so that even
with certain detection the discount factor must be at least $38/45
\approx 0.844$. The parameters are illustrative.

\paragraph{Status of the condition:} Condition \eqref{eq:Ei} is an
inequality in the model's parameters. It holds when the revenue at
stake is large relative to the hourly operating cost. When it fails
for a firm, that firm gains by restraining alone against expanding
rivals, and universal expansion is not an equilibrium; this does not
mean that restraint is dominant for that firm. Historical evidence is
consistent with the condition but does not establish it across market
structures and firm sizes: manufacturing plants consistently respond
to productivity improvements by expanding output rather than reducing
labor input \citep{baily1992productivity}; even during periods of
declining work hours (1890s--1940s), reductions came primarily through
labor movement pressure and legislation rather than voluntary firm
decisions \citep{costa2000wage}; and \citet{schor1992} documented a
pattern of overwork. The mechanism is that of a contest for shares
\citep{tullock1980}.

For workers, universal expansion forgoes the leisure gain $\beta\,
(\alpha-1)H_{i}/\alpha$ that restraint would give, as the next
proposition records.

\begin{proposition}[Pareto-Inefficiency in Productivity Competition]
\label{prop:pareto-productivity}
With a common productivity factor $\alpha > 1$, positive hours $H_{i} >
0$, $c > 0$ and $\beta > 0$, Universal Cooperation ($s^C$) strictly
Pareto-dominates Universal Defection ($s^D$) for the modeled firms and
their workers.
\end{proposition}

\begin{proof}
Under universal expansion and under universal restraint firm $i$ has
the same market share, $Q_i/\sum_j Q_j$, because the common factor
$\alpha$ cancels. Thus $U_{F,i}(s^C) - U_{F,i}(s^D) = c\,H_i\,(1 -
1/\alpha) > 0$. At fixed pay the workers' difference is $U_{W,i}(s^C)
- U_{W,i}(s^D) = \beta\,H_i\,(1 - 1/\alpha) > 0$. Both inequalities
hold for every $i$, because $\alpha > 1$ and $c$, $\beta$ and $H_i$
are positive, so $s^C$ strictly Pareto-dominates $s^D$ for all modeled
firms and workers.
\end{proof}

\paragraph{Note on Market Revenue:} The proof assumes that total market
revenue $r$ is constant. In practice, universal output expansion would
likely lower prices; whether that lowers or raises the firms' total
revenue depends on demand, and lower prices would benefit buyers. The
model is a partial-equilibrium comparison of firms and their workers,
in which workers' money pay and what it buys are held fixed. If
workers are also consumers of the output, changes in prices and
purchasing power affect the comparison of consumption with leisure;
the model does not determine the economy-wide welfare trade-off behind
Keynes's fifteen-hour prediction.

\paragraph{Result (under the stated condition):} Combined with
the N-player structure, externalities, and enforcement obstacles
discussed in Section \ref{sec:productivity-mapping}, productivity competition
in this stylized model is a structural tragedy at the level of the
stage game whenever \eqref{eq:Ei} holds strictly for every firm:
universal expansion is then the unique equilibrium and is
Pareto-dominated by universal restraint. The condition is also
necessary for uniqueness. Firms may differ in size and in
initial hours; the condition states how these enter, through each
firm's hours and market share. The result depends on the stated
assumptions: a common productivity factor, fixed revenue and fixed
pay.

\subsection{Formalization: Extensions, Robustness, and Limitations}

The extensions in this subsection are not derived formally; they are
hypotheses about how the stylized result may generalize.

\subsubsection{Extension to Multiple Welfare Dimensions}

The same structural logic may apply to other dimensions of worker
welfare: when productivity increases, firms can raise wages
proportionally, enhance benefits or stabilize employment (C), or
maintain them while expanding output (D). We conjecture that each
dimension exhibits the same five conditions for structural
tragedy. Which dimension manifests the coordination failure would then
depend on institutional constraints. Hour regulations shift pressure
to wages; wage floors shift pressure to employment; strong employment
protections shift pressure to benefits. The declining labor share of
income globally \citep{karabarbounis2014global, elsby2013decline} is
consistent with the hypothesis that productivity gains flow away from
worker welfare across dimensions, though the specific manifestation
varies by institutional context.

\subsubsection{Relaxing the Uniform Productivity Assumption}

The assumption of a common gain $\alpha$ is most realistic for
economy-wide technological shifts (e.g., electrification,
computerization, internet) that eventually diffuse broadly, or for
long-run analysis after differential adoption rates have equilibrated
through competitive selection.

\paragraph{Dominance Under Heterogeneity:} With heterogeneous
productivity gains $\alpha_i$ varying across firms, the exact
condition of Proposition \ref{prop:dominance} is not derived here. The
sufficient condition of Theorem \ref{thm:main}(b) still applies with
$\alpha_i$ in place of $\alpha$: expansion is dominant for firm $i$
whenever its market-share gain exceeds its added operating cost at
every competitor profile. Where that holds for every firm, which is
not shown here, the individual incentive to defect regardless of what
others do persists under heterogeneous productivity shocks.

\paragraph{Pareto-inefficiency Under Heterogeneity:} Proposition
\ref{prop:pareto-productivity} assumes a common gain. With
heterogeneous $\alpha_i$, workers remain strictly better off under
$s^C$, but the outcome for firms depends on the distribution of
productivity gains. When heterogeneity is extreme, high-productivity
firms may prefer universal defection to capture market share from
rivals. Such heterogeneity leads to rapid market consolidation through
the exit of low-productivity firms \citep{melitz2003}; we conjecture
that it does not eliminate the tragedy but localizes it to the
surviving firms of similar productivity. For the empirically common
case of moderate heterogeneity within industries, we expect
Pareto-inefficiency to hold: the productivity distribution of
surviving firms exhibits substantial but not extreme dispersion
\citep{syverson2011}. That all firms then benefit from coordinated
restraint relative to the racing equilibrium requires additional
assumptions. Where the dominance property (Condition 3) holds,
coordination remains difficult even when C4 holds weakly, because
\textit{every} firm prefers unilateral defection.

\subsubsection{Generality and Robustness}

\paragraph{Beyond Binary Choices:} The binary choice is a
simplification. In a continuous version in which firms choose specific
hours $h_i \in [0, H_{max}]$, there is in general no dominant action,
because the best choice of hours depends on rivals' output. We
conjecture that equilibrium hours in such a model
exceed the level that firms and workers would jointly prefer when
competition for market share is sufficiently intense.

\paragraph{Risk Aversion:} Under uncertainty about future market
positions, exposure to severe downside losses may strengthen
incentives to expand. Whether greater risk aversion does so depends on
the distributions of outcomes under expansion and restraint.

\paragraph{Dependence on Market Structure:} How the result depends on
market structure is stated by condition \eqref{eq:Ei} in Proposition
\ref{prop:dominance} and, for equal output per hour, by Corollary
\ref{cor:equal-productivity}; the comparison of Proposition
\ref{prop:pareto-productivity} does not depend on firm sizes.

\subsubsection{Gains from Restraint Increase with Productivity}

In the model, the gains from universal restraint increase with the
productivity factor. For firm $i$ the gain is $c H_i (1 -
1/\alpha)$, and for its workers the leisure gain is $\beta H_i (1 -
1/\alpha)$; both increase in $\alpha$, toward $c H_i$ and $\beta
H_i$. Adding the two uses the model's common units and is not a
measure of social welfare. This is a comparison of payoffs between the
two unanimous outcomes, not a proved limit of an equilibrium loss:
whether universal expansion is the equilibrium depends on condition
\eqref{eq:Ei}, which also changes with $\alpha$. Where the condition
holds, the comparison suggests that the gap between potential leisure
and actual leisure widens with productivity: economies whose
productivity has grown severalfold since Keynes wrote could provide
far more leisure than 1930s economies.

\subsection{Empirical Assessment: The European Case}
\label{sec:europe}

The European experience offers an informative comparison. The
discussion in this subsection is interpretive: it applies the stylized
model where its condition holds and offers hypotheses beyond it. If
our structural analysis is correct, under favorable institutional
conditions we would expect: (1) partial mitigation achieving some
welfare improvements but not full Keynesian leisure, (2) high costs
requiring substantial institutional apparatus, and (3) continuing
competitive pressure on these arrangements over time.

\paragraph{Favorable Conditions:}

Europe is among the most favorable cases for coordination: wealthy
democracies with strong institutions, supranational governance
(European Union) providing coordination capacity, cultural consensus
supporting work-life balance, strong labor unions with political
influence, and legal frameworks mandating maximum hours and minimum
vacation. If coordination is feasible under present conditions,
Europe is where we would expect it.

\paragraph{Partial Mitigation:}

Reported average annual hours per worker in 2024 were about 1,330 in
Germany and about 1,800 in the United States \citep{oecd2025}. The
OECD cautions that differences in sources and methods make these data
unsuitable for comparing country levels. We therefore treat the
figures as descriptive context alongside the institutional evidence,
not as a measure of the effect of coordination. Even German hours
remain far above the 780 that correspond to Keynes's fifteen-hour
week. The coordination is partial, not complete.

Moreover, productivity-welfare decoupling manifests differently across
institutional contexts. While European workers gained leisure relative
to US counterparts, Europe also exhibits declining labor share of
income since 1980, slower wage growth relative to productivity,
persistent pressure on social benefits, and higher unemployment in
several major economies. One hypothesis is that the tragedy is not
escaped but redistributed: European workers gained leisure at the
cost of wage growth and employment security, while US workers
experienced stagnant wages despite longer hours. On this hypothesis,
both regions show productivity-welfare decoupling driven by
competitive pressure, manifesting through whichever dimension
institutional constraints leave exposed.

\paragraph{High Costs:}

This coordination comes at substantial cost.

\textit{Economic costs:} European GDP growth has lagged the United
States over recent decades, with significant divergence since 2008
\citep{worldbank2024}. This divergence is driven by multiple
factors. Shorter working hours lower total output for a given output
per hour, but we do not attribute the growth gap to them.

\textit{Regulatory costs:} Maintaining shorter hours requires
extensive regulation: France's 35-hour workweek involves complex rules
about overtime, exceptions, and compliance monitoring
\citep{leglobal2024}, while Germany's system requires workplace
councils, sector-level bargaining, and intricate labor law.

\textit{Competitive pressure:} European firms face competitive
disadvantages against firms from regions with longer working hours,
generating continuous political pressure to loosen regulations.

\paragraph{Pressure Over Time:}

Coordination also shows signs of pressure. France's 35-hour law
contains numerous exceptions and workarounds; actual hours often
exceed nominal limits \citep{leglobal2024}. In the data of
\citet{huberman2007}, annual hours in most Old World countries
continued to fall between 1980 and 2000 (Sweden's rose), while US
hours were about flat. We therefore do not claim that the European
decline has stopped; our hypothesis is that it continues under
persistent competitive pressure and at the costs described above.

\paragraph{Interpretation:}

In our assessment, institutional intervention achieves real but partial
and costly mitigation, not elimination of the tragedy; whether
institutions can reach the cooperative outcome without structural
change to competitive conditions themselves is, in our view,
doubtful. This is consistent with the hypothesis that the underlying
structure persists when it is overridden; it does not establish that
hypothesis.

\subsection{Standard Solutions, Structural Barriers, and Possible Interventions}
\label{sec:productivity-solutions}

Our framework suggests why commonly proposed solutions to the
productivity-welfare decoupling face structural barriers, and it
points to interventions that could address the productivity tragedy by
changing one or more of the five conditions. The statements in this
subsection are implications of the stylized model where its condition
holds, and hypotheses beyond it. These are conditional mechanisms and
assessments, not proofs that the proposals fail; changing a condition
is neither shown necessary nor sufficient for successful coordination.

\paragraph{Standard proposals:} For each commonly proposed solution,
the following table notes which condition it does not by itself remove
and the obstacle this leaves.

{\setlength{\tabcolsep}{5pt}
\renewcommand{\theHtable}{proposals.productivity}
\begin{longtable}{@{} Q{3.9cm} Q{12.1cm} @{}}
\toprule
\textbf{Proposal} & \textbf{Incentive or enforcement obstacle that remains} \\
\midrule
\endhead
\bottomrule
\endlastfoot
Individual firm action: firms voluntarily reduce working hours or
raise wages in proportion to productivity gains & The incentive to
expand when competitors expand, which in the model is strictly
positive when \eqref{eq:Ei} holds strictly: a firm that reduces hours
alone loses market share to its competitors. Voluntary restraint by
individual firms is unlikely to succeed when defection yields higher
payoffs regardless of competitors' strategies \citep{olson1965}. \\
\addlinespace
Industry-level voluntary agreements among firms on reduced hours or
higher wages & The enforcement obstacles (C5). As \citet{stigler1964}
argued for cartels, a firm can gain by secretly expanding hours to
capture market share while others comply; in the model every firm with
at most half the market has this incentive when \eqref{eq:Ei} holds
strictly. Reputation and industry associations give weak enforcement;
without credible sanctions we expect such agreements to be fragile. \\
\addlinespace
National regulation: labor standards requiring reduced working hours
or minimum wage increases tied to productivity & International
competition and capital mobility: production may shift to less
regulated jurisdictions, the race-to-the-bottom dynamic
\citep{rodrik1997}. In Finland, despite strong labor market
institutions and centralized wage bargaining, exposure to
international trade systematically drove exits of high-labor-share
plants and reduced the aggregate labor share
\citep{bockerman2012globalization}, which suggests that even wealthy
democracies with strong institutions face this pressure. \\
\addlinespace
International labor agreements: treaties on maximum working hours and
minimum wages tied to productivity & The enforcement obstacles
(C5). The International Labour Organization lacks sanctioning power,
compliance is voluntary, and nations can endorse standards publicly
while maintaining lax enforcement to attract capital investment
\citep{flanagan2003}. Effective governance requires graduated
sanctions and low-cost monitoring \citep{ostrom1990}, both absent at
the international level due to sovereignty concerns and international
anarchy. \\
\addlinespace
Technological unemployment: automation will eventually force shorter
working weeks as human labor becomes unnecessary & Misidentifies the
mechanism. Employment has remained robust because technology creates
complementary tasks even as it automates existing ones
\citep{autor2015}. Automation increases productivity, which can raise
the gain from universal restraint while leaving expansion as the
unique stage-game equilibrium, as long as \eqref{eq:Ei} remains
strict. Since Keynes wrote, productivity has grown severalfold, while
annual working hours in the United States have fallen by far less. \\
\addlinespace
Cultural change: a shift of societal preferences toward leisure over
consumption & Insufficient against structural pressure. Despite
expressed preferences for more leisure time, American workers face
structural pressures that override these preferences
\citep{schor1992}; where \eqref{eq:Ei} holds strictly in the model,
preferences alone do not overcome the pressure to maintain competitive
position when others work longer hours. \\
\end{longtable}
\addtocounter{table}{-1}
}

We expect partial solutions to achieve at most partial and costly
coordination under continuing pressure, which is how we read the
European experience.

\paragraph{Possible structural interventions:} Five interventions would each change one
of the conditions. We present them as possible structural intervention
targets, not policy recommendations, to clarify the magnitude of
barriers any solution must overcome. The following table states each
intervention and the difficulty it meets.

{\setlength{\tabcolsep}{5pt}
\renewcommand{\theHtable}{interventions.productivity}
\begin{longtable}{@{} Q{2.5cm} Q{5.3cm} Q{7.85cm} @{}}
\toprule
\textbf{Condition to change} & \textbf{Intervention} & \textbf{Difficulty} \\
\midrule
\endhead
\bottomrule
\endlastfoot
Condition 1 (N-Player Structure) & Consolidate decision-makers into
smaller groups, making coordination feasible through reduced
free-rider problems: as \citet{olson1965} established, smaller groups
dramatically lower coordination costs. & Industrial consolidation
meets antitrust law in most jurisdictions, and for good reason:
consolidation that enables coordination on labor practices would
simultaneously enable price-fixing and output restriction, trading one
market failure (coordination tragedy) for another (monopoly power)
that harms consumer welfare. International consolidation of labor
regulation authority could reduce the number of competing
jurisdictions, but it faces political barriers from sovereignty
concerns and competing national interests. \\
\addlinespace
Condition 2 (Binary Choice with Externalities) & Eliminate the
competitive externality where one firm's expansion harms competitors'
market positions. & We regard this as infeasible in a market economy,
since competitive pressure is inherent to it. As
\citet{schumpeter1942} argued, competition through creative
destruction, where successful firms displace less efficient ones, is
capitalism's fundamental mechanism for driving productivity
growth. The externality is reciprocal: protecting compliant firms from
market share loss necessitates restricting expanding firms, which
contradicts the productivity-enhancing function of market
competition. \\
\addlinespace
Condition 3 (Dominance Property) & Reduce the existential stakes of
restraint so that cooperation does not threaten firm or worker
survival. For workers, mechanisms include Universal Basic Income
(UBI), stronger unemployment insurance, or other social safety nets
that provide unconditional material security. As
\citet{vanparijs1995} argues, guaranteeing an unconditional material
baseline grants workers ``real freedom'' to refuse exploitative
employment terms. For firms, mechanisms would include subsidies, tax
incentives, or market structures that reward restraint rather than
penalizing it; for example, government contracts preferentially
awarded to firms with reduced hours, or subsidies compensating for
reduced output during hour reductions. & Fiscally adequate UBI levels
require substantial tax increases or redistribution that face
political opposition across diverse political systems
\citep{martinelli2017}, and workers with basic income still face
relative income competition and status concerns that perpetuate long
hours \citep{frank1999}. The mechanisms for firms face moral hazard
(firms might claim hour reductions while maintaining actual hours),
substantial ongoing public expense, and competitive distortions
favoring subsidized over non-subsidized firms. Most fundamentally, any
mechanism rewarding restraint must overcome the problem that market
structures inherently reward expansion, which returns to the
difficulty of changing Condition 2. \\
\addlinespace
Condition 4 (Pareto-Inefficiency) & Alternative assessment: accept
the equilibrium as efficient rather than seeking coordination. Some economists, notably
\citet{prescott2004}, argue that transatlantic differences in working
hours reflect rational responses to marginal tax rate differences
rather than coordination failures, with Americans optimally choosing
longer hours given higher after-tax wages. & As \citet{pencavel2018}
argues, observed hours often do not reflect worker preferences;
instead, institutional and market constraints significantly determine
working time. The tax-incentive explanation also cannot account for
the secular decline in labor share within countries over time,
including during periods of tax stability, and the micro-level
restructuring documented by \citet{bockerman2012globalization} is
consistent with a coordination failure operating alongside any tax or
preference effects. If workers genuinely prefer current hours and the
equilibrium is Pareto-efficient, no coordination problem exists. Our
framework and the empirical evidence suggest otherwise. \\
\addlinespace
Condition 5 (Enforcement Difficulty) & Solve the enforcement problem
through effective global labor standards with comprehensive monitoring
and binding sanctions. This requires international agreement on
maximum hours and minimum wage-productivity linkages, verification
mechanisms for compliance across jurisdictions, sanctions sufficient
to deter defection, and sufficient participation to prevent
race-to-bottom through non-compliant jurisdictions. & The obstacles
are those noted above for international agreements, together with
capital mobility that enables regulatory arbitrage: achieving binding
global labor standards would require unprecedented coordination among
competing nations, effectively solving a higher-order coordination
problem to address the first-order productivity tragedy. \\
\end{longtable}
\addtocounter{table}{-1}
}

\paragraph{Comprehensive Solution:} In our view the productivity tragedy
can be addressed, but only through interventions that address its
structural foundations. Given the difficulty of breaking any single
condition completely, the most feasible approach likely combines
partial interventions across multiple conditions: strengthening social
safety nets (C3), enhancing coordinated regional labor standards (C5)
among economically integrated blocs, and accepting modest efficiency
costs from reduced competitive pressure. As we read the European
experience, even such comprehensive efforts face continuing global
competitive pressure, and we expect incremental reforms that leave the
five conditions substantially intact to achieve at most partial and
costly coordination.

\section{The Tragedy of AI Governance}\label{sec:ai}

We now apply the unified framework to AI development, arguing that AI
governance exhibits the same five conditions, with higher assessed
intensity. We argue that this makes coordination more difficult than
for historical arms control cases including NBC weapons.

The phrase ``the tragedy of AI governance'' was adopted independently
in developing this paper; \citet{chesterman2023tragedy} had previously
published under the same title. He uses it for a mismatch of leverage
and interest: in his words, ``those with the greatest leverage to
regulate AI have the least interest in doing so, while those with the
greatest interest have the least leverage,'' which he attributes to
the shift of power from public to private hands, states' wariness of
regulating too aggressively, and the dysfunction of global processes
for managing collective action problems. These overlap with the
barriers discussed here. We use ``tragedy'' in the structural sense of
Section \ref{sec:framework}: how individually rational choices can
produce a collectively inferior outcome and how enforcement obstacles
can impede cooperation.

\paragraph{Scope of Analysis:} Our analysis focuses on the governance
challenge posed by transformative AI breakthroughs: systems
approaching or exceeding human-level general intelligence (AGI) or
superintelligence. We assume such systems pose non-excludable
existential risks if developed without adequate safety measures
\citep{bostrom2014, russell2019, bengio2025}. We examine the implications
of relaxing this assumption after arguing for Pareto-inefficiency
(C4).

\paragraph{Status of the Claims:} Here $C$ means adhering to a given
safety or restraint protocol and $D$ means departing from it. In the
notation of Section \ref{sec:formal-result}, universal departure is a
Pareto-dominated equilibrium (level 1) exactly when no actor gains by
adhering alone, $g_i(s^D_{-i}) \ge 0$, and every actor prefers
universal adherence, $v_i > 0$. It is the only equilibrium (level 2)
if departing is better for each actor whatever the others do
(Condition 3); for two identical actors this holds exactly when $T >
R > P > S$, the prisoner's dilemma ordering. If instead no actor gains by departing alone,
$G_i \le 0$, universal adherence is an equilibrium as well and the
problem is one of assurance; this does not show that these are the
only two equilibria. Without Condition 3, Theorem \ref{thm:main} does
not establish uniqueness. With more than two actors, comparing only
the two uniform profiles does not establish Condition 3. Each
comparison must count the benefits, costs and risks of both outcomes
for every actor. We give informal arguments that Conditions 3 and 4
are plausible for frontier AI development; we do not establish them in
general. In the formal
model of \citet{abraham2025prisoners} the corresponding game is a
prisoner's dilemma for only part of the parameter range and a
coordination (assurance) game elsewhere, and
\citet{katzke2024manhattan} argue that a race to superintelligence is
a trust dilemma, not a prisoner's dilemma. The comparison with
arms control and the policy discussion below are conditional on these
conditions and are the author's qualitative assessment.

\subsection{Mapping the Five Conditions to AI Development}
\label{sec:ai-mapping}

\subsubsection{Condition 1: N-Player Structure}

AI development involves numerous independent actors: AI labs, tech
companies, universities and research institutions, nations, and
open-source communities. Their number is large and growing, and no
single actor can coordinate AI development unilaterally. This count,
and the High grade for Condition 1
in Table \ref{tab:intensity-comparison}, includes developers,
deployers and states, not only the few laboratories at the
frontier. Counting them together is a grading choice; it does not
establish that one common payoff model applies to all of these
actors.

\subsubsection{Condition 2: Binary Choice with Externalities}

When AI capabilities advance, each actor faces essentially a binary choice:
\begin{itemize}
    \item Strategy $C$ (Cooperate/Restrain): Slow development to
      ensure safety, conduct thorough testing, wait for governance
      mechanisms
    \item Strategy $D$ (Defect/Advance): Accelerate capabilities
      development to maintain competitive position
\end{itemize}

Individual choices create negative externalities. When one actor
accelerates development, competitors face pressure to accelerate or
lose market or strategic position, safety margins narrow as
capabilities race ahead of governance, and the risks from powerful AI
systems affect everyone regardless of who develops them. The
externality is severe since one actor's defection substantially
increases risks for all.

\subsubsection{Condition 3: Dominance Property}

We argue that acceleration is plausibly dominant, that is, preferred
regardless of competitors' choices:

\textit{Case 1—If competitors restrain:} The actor gains a decisive
first-mover advantage by accelerating. In AI markets characterized by
network effects, data advantages, and talent concentration,
first-mover advantages are enormous. The actor that reaches advanced
capabilities first may capture the entire market.

\textit{Case 2—If competitors accelerate:} The actor must accelerate to
avoid strategic elimination. Falling behind in AI capabilities may
mean permanent market displacement for firms or strategic
vulnerability for nations. When competitors accelerate, restraint
means obsolescence.

On this argument the dominance is driven by existential stakes for
firms, nations, and researchers. For example, a firm that cannot raise
capital for massive compute (or acceleration) risks obsolescence. As an
illustration of this pressure, on our reading, even the leading
non-profit actor restructured first into a `capped-profit' entity
\citep{openai2019lp} and eventually into a Public Benefit Corporation
to remove fundraising limits \citep{openai2025structure}.

\subsubsection{Condition 4: Pareto-Inefficiency}

We argue that universal restraint Pareto-dominates universal
acceleration. If all actors coordinate to slow development and
prioritize safety, risks are reduced through careful development,
relative competitive positions are preserved, time is available for
governance mechanisms to develop, and the benefits of AI still
accrue. Under mutual acceleration (the equilibrium if Condition 3
holds), risks accumulate faster than governance mechanisms develop,
safety testing is compressed or skipped, first-mover races waste
resources on competitive positioning, and catastrophic outcomes that
benefit no one become possible.

Many AI researchers regard the current trajectory as collectively
harmful \citep{askell2019}: in surveys, substantial minorities
estimate non-trivial probabilities of extremely bad outcomes,
including human extinction \citep{bengio2025, grace2018}. Yet
development accelerates.

\paragraph{Asymmetric Gains and the Winner-Take-All Problem:} Unlike
productivity competition, where firms keep roughly proportional
positions, an AI capabilities race may be winner-take-all: the first
to achieve advanced AI may capture decisive market or strategic
advantages. This asymmetry threatens the Pareto-inefficiency
condition. An actor who expects to win might rationally prefer the
race, with its enormous gains, to a cooperative outcome in which
advantages are constrained or shared; for a prospective winner the
situation appears not as a tragedy but as a rational pursuit of
dominance.

However, we argue that three factors restore the tragedy structure
even for prospective winners:

\paragraphit{First, existential risk is non-excludable.} The catastrophic
risks from misaligned or insufficiently robust advanced AI systems
affect all actors regardless of who develops them first
\citep{bostrom2014}. A winner who achieves advanced AI through hasty
development faces the same risks of loss of control, unintended
consequences, or catastrophic failure as any other actor, and by the
instrumental convergence thesis the subgoals such systems pursue may
conflict with human interests regardless of which actor deploys
them. Winners and losers alike face potential extinction or permanent
disempowerment.

We argue that in a race the rate of capability gains tends to outpace
the rate of safety verification. If so, the winner may well deploy an
unsafe system before establishing control, which could make the
expected utility of `winning' negative even for the victor. The survey
evidence cited above bears on this. Accurate risk estimates are not
enough by themselves:
Condition 4 requires that, for every actor, the expected payoff under
universal restraint exceed the expected payoff under universal
acceleration, counting the benefits, costs and risks of both (see the
critical caveat below). If these probability estimates are even
roughly accurate and that comparison holds, the race is collectively
inferior under this payoff comparison, even though accelerating can
remain individually rational, which is the dilemma this paper
studies.

\paragraphit{Second, winner identity is uncertain ex ante.} While racing,
no actor knows that it will be the winner: numerous sophisticated
actors, corporate and national, have credible capabilities. The
``DeepSeek Shock'' of early 2025 illustrates this. Its reported
reasoning performance and its release of smaller distilled models
suggested that capability advantages can diffuse, which challenges the
durability of incumbent leads \citep{deepseek2025}. To the extent that
the final breakthrough remains winner-take-all, such diffusion can
make current leads less secure. From behind this veil of uncertainty, all actors, even
current leaders, may face expected losses from an uncoordinated race.

This parallels arms races where no state knows ex ante whether it will
achieve dominance. Even if one state would retrospectively prefer the
outcome where it achieved hegemony, we argue that, ex ante, states
would prefer coordinated restraint to a race with uncertain winners
and substantial risks, provided the risk-adjusted comparison of
Condition 4 holds for each of them (see the critical caveat below).

\paragraphit{Third, even decisive technological leads may be temporary.}
Unlike traditional military conquest where territory captured remains
under control, AI capabilities can potentially be copied, stolen, or
independently rediscovered. Winner-take-all assumptions may prove
mistaken if multiple actors achieve advanced AI within relatively
short timeframes, or if AI capabilities prove impossible to monopolize
due to espionage, independent discovery, or proliferation.

We therefore argue that, even with extremely asymmetric potential
gains, the Pareto-inefficiency condition is plausible when existential
risks are incorporated into each actor's calculation. On this
argument the racing outcome threatens all actors including prospective
winners, making coordinated restraint collectively preferable despite
competitive temptations.

\paragraph{Critical caveat:} This defense of Pareto-inefficiency relies
on two empirical claims: (1) catastrophic risks from advanced AI are
sufficiently probable and severe to dominate even winner-take-all
competitive gains, and (2) actors incorporate these risks into their
decision-making rather than pure output or power maximization. The
first claim is motivated by arguments about catastrophic AI risk
\citep{bostrom2014, russell2019}; this paper does not estimate the
relevant probabilities or payoffs. The second claim
is more tenuous: if actors systematically discount existential risks or
believe they can control advanced AI despite expert warnings, the
existential dimension of the tragedy weakens.

However, even under systematic discounting of existential risk, a
coordination failure may remain at lower stakes. It
would manifest through wasteful resource competition: duplicated R\&D
spending across labs, unsustainable talent bidding wars, massive
capital obligations against uncertain returns, and infrastructure
investments vulnerable to algorithmic advances that devalue
hardware moats. Such
a costly contest is Pareto-inefficient only if all participants would
prefer coordinated investment, which duplicated spending alone does
not show. As an illustrative profile for this scenario we rate C3* and
C4 High instead of Extreme, assuming severe competitive pressure and
major but recoverable losses; the enforcement difficulty (C5) remains
Extreme, as the technical barriers to verification are unchanged.
Under this profile, (H, H, H, H, E), the Index is about 1,585. Whether
the tragedy persists in this scenario depends, as before, on the
stage-game inequalities and on the persistence hypothesis.

\subsubsection{Condition 5: Enforcement Difficulty}

Four structural barriers obstruct coordination; Section
\ref{sec:ai-comparison} compares them with historical arms control.
\textit{International anarchy:} no global authority can enforce
coordinated AI governance, nations compete for AI advantages, and
international agreements lack binding force. \textit{Verification
difficulty:} AI development occurs in software, model weights can be
copied and modified, and training runs can occur
covertly. \textit{Monitoring infeasibility:} what constitutes
`dangerous' AI capabilities may be ambiguous until after development,
and research is distributed across many institutions. \textit{Commitment
problems:} each actor has a strong incentive to defect and gain
first-mover advantages once others restrain, secret defection is
technologically feasible, and the first-mover disadvantage from
restraint is severe.

\subsubsection{Result}

On the arguments above, AI development exhibits the five conditions,
with high assessed intensity when catastrophic risks are incorporated
into actors' assessments. If Conditions 3 and 4 hold, then by Theorem
\ref{thm:main}(b) universal acceleration is the unique Nash
equilibrium of the stage game despite being Pareto-dominated by
coordinated restraint. If only the level 1 inequalities hold, it is an
equilibrium but need not be the only one. The statement on the status
of the claims at the start of this section applies to these
conditions. The
empirical evidence suggests that leading actors do recognize
catastrophic risks even while competitive pressure compels continued
development, which is the pattern our framework describes.

\paragraph{Robustness of the Racing Dynamic (conjecture):}

We conjecture that, as with productivity, this structure is robust to
relaxed assumptions. In a continuous version in which actors choose a
level of investment or compute $k_i \in [0, K_{max}]$, the marginal
strategic advantage of increasing $k_i$ may remain positive as long as
competitors are advancing; in that case `slowing down' would not be a
stable outcome and investment would be pushed toward the maximum
feasible level. In formal race models the
outcome depends on the parameters. Furthermore, fear of falling
behind may amplify the race. If the loss of
sovereignty or market relevance is viewed as a certain, immediate
catastrophe, while coordination failure consequences (whether
existential or economic) are viewed as probabilistic and future
catastrophe, actors may prioritize avoiding the immediate
deterministic loss, thereby intensifying the acceleration.

\subsection[AI Governance vs. Historical Arms Control: An Eight-Dimension Comparison]{AI Governance vs.\ Historical Arms Control:\\ An Eight-Dimension Comparison}
\label{sec:ai-comparison}

To assess whether AI governance is more or less difficult than
previous coordination challenges, we compare AI with NBC weapons, which all achieved substantial
(though imperfect) international coordination. Table
\ref{tab:ai-comparison} summarizes the comparison across eight
dimensions; the entries are the author's qualitative assessments. The
paragraphs that follow give the reasons for the entries, taking
related dimensions together.

\begin{table}[!ht]
\centering
\small
\caption{AI Governance vs. Historical Arms Control Across Eight
  Dimensions (the author's qualitative assessment). In our assessment
  AI compares less favorably for coordination on most dimensions,
  whereas historical cases benefitted from mitigating factors (e.g.,
  observability, separability) in key areas.}
\label{tab:ai-comparison}
\newcolumntype{Y}{>{\raggedright\arraybackslash\hyphenpenalty=10000\exhyphenpenalty=10000}X}

\begin{tabularx}{\textwidth}{@{} p{3cm} Y Y Y @{}}
\toprule
\textbf{Dimension} & \multicolumn{1}{c}{\textbf{Chemical / Bio.}} & \multicolumn{1}{c}{\textbf{Nuclear}} & \multicolumn{1}{c}{\textbf{AI}} \\
\midrule
Strategic Value & Mainly military, limited & Military, high & Military + Economic, extreme \\
\addlinespace
Convergence & Mainly military; large civil industry & Mainly military; civil power sector & Military + Economic \\
\addlinespace
Dual-Use & Limited separability & Limited separability; enrichment partly observable & Inseparable \\
\addlinespace
Verification & Chemical: inspected facilities; biological: no verification regime & Observable (materials) & Difficult: software hard to observe; large-scale compute partly observable \\
\addlinespace
Reputational Costs & High (taboo) & Mixed (stigma and prestige) & Low (prestige) \\
\addlinespace
Substitutes & Yes (conventional) & Partial & None \\
\addlinespace
Entry Barriers & Moderate & Very High & High at the frontier; lower for replication and deployment \\
\addlinespace
Time Scale & Decades to a comprehensive treaty (observed) & Decades to a comprehensive treaty (observed) & Possibly shorter pre-deployment window (conjectured) \\
\bottomrule
\end{tabularx}
\end{table}

\paragraph{Strategic value and convergence:} Chemical and biological
weapons have limited strategic value: they are unreliable, difficult to
control, and often militarily ineffective compared to conventional
weapons, which reduced incentives for development and made restraint
easier. Nuclear weapons have high strategic value for deterrence,
making complete elimination difficult, but their value is mainly
military, and economic incentives for developing them are weak. All
three were developed for military advantage, so the coordination
problem for the weapons was mainly a security problem, although the
underlying chemical and biological industries and civilian nuclear
power are large civilian sectors, separable in part from the weapons.
AI has extreme strategic value in both domains: military applications
include autonomous weapons, intelligence analysis, cyber operations,
and logistics optimization, and economic applications include
productivity enhancement, automation, market advantage, and platform
dominance. In our assessment the two imperatives converge almost
completely: every advance in civilian AI capabilities also enhances
military potential. Even if nations could overcome military
competition (as partial nuclear arms control achieved), economic
competition would independently drive AI development; even if firms
could coordinate to slow capabilities races for safety, military
imperatives would independently accelerate it. We expect this dual
reinforcement to make coordination more difficult than in previous,
mainly single-domain challenges.

\paragraph{Dual use and verification:} Many industrial chemicals and
much civilian biology are dual-use, and distinguishing civilian work
from weapons programs is difficult \citep{vaynman2023dualuse}.
Chemical weapons production can be inspected at declared facilities
under the Chemical Weapons Convention, whereas the Biological Weapons
Convention has no verification regime \citep{koblentz2009,
  vaynman2023dualuse}. Civilian and military nuclear programs share
facilities and technology, but nuclear facilities and material flows
provide physical targets for safeguards and verification. For AI the same models,
training techniques, and hardware serve civilian and military
purposes. In our assessment, attempting to restrict military AI while
permitting civilian AI is extremely difficult: under any such
agreement developers could present their research as civilian, and the
distinction itself would be hard to draw. AI development also occurs
in software and is difficult to verify: model weights can be copied
instantly, training runs can be concealed, and testing for dangerous
capabilities requires developing the very systems one seeks to
restrict. Large-scale training does depend on specialized chips and
data centers, which are partly observable \citep{sastry2024}.

\paragraph{Reputational costs, substitutes and entry barriers:}
Chemical and biological weapons face powerful taboos with strong
international norms \citep{price1995, opcw2025, unoda2025}, and they
have conventional substitutes for most military objectives. Nuclear
weapons face mixed stigma and prestige and have only partial
substitutes, since conventional forces do not fully replace nuclear
deterrence; their very high entry barriers (billions of dollars,
advanced technology, years of development) concentrated capabilities
in few hands, enabling coordination. For AI the reputational incentive
is reversed: leading development confers prestige, not stigma, so that
restraint can appear as weakness. In our assessment AI has no close
substitutes for either economic or military applications, so actors
who restrain risk falling behind, which strengthens the incentive
described under Condition 3. AI entry barriers are high at the
frontier, where training costs are very large and concentrated in few
actors, but lower and declining for the replication and deployment of
existing capabilities, which multiplies the number of relevant actors
and the coordination problem.

\paragraph{Time scale:} Chemical and nuclear weapons had decades
between emergence and comprehensive arms control: large-scale chemical
use in 1915 was followed by the 1925 Geneva Protocol banning use and
by the 1993 Chemical Weapons Convention, and the first nuclear use in
1945 by control proposals within a year and by the Non-Proliferation
Treaty in 1968. In both cases the treaties followed the use of the
weapons, and governments could learn from that use. For transformative
AI, on the assumption of this section, governance would have to be in
place before the capability is deployed. The relevant window therefore
runs from credible evidence of dangerous capabilities to their
deployment, and we conjecture that it is shorter than the historical
intervals. Expert forecasts do not settle this: in a survey conducted
in 2023 the aggregate forecast, conditional on science continuing
undisrupted, gave a 50\% chance of AI outperforming humans at every
task by 2047, thirteen years earlier than a similar survey had given a
year before \citep{grace2024}, but neither this forecast nor its
revision estimates that window. The two sides are also not measured on
the same basis: the historical intervals are observed
treaty-development times, whereas the AI window is a conjecture about
a period that has not yet run.

\paragraph{Summary across dimensions:} In our assessment, AI compares
less favorably for coordination than successful historical arms
control on most of the eight dimensions of Table
\ref{tab:ai-comparison}. The comparison is a qualitative assessment,
not a ranking established by measurement. In our view the
inseparability of economic and military AI (the dual-use dilemma) is
the most important barrier: restricting military AI would require
restricting economic AI, at a large cost to a nation's productivity.

\subsection{Empirical Assessment: The Russia-Ukraine Drone War}

The Russia-Ukraine conflict illustrates competitive dynamics under
existential stakes. Drones had been used in the conflict since 2014,
but after Russia's full-scale invasion in February 2022 their
production and use escalated rapidly \citep{csis2024, watling2023}:
about 2.2 million drones were produced in Ukraine in 2024, and more
than 4.5 million were forecast for 2025 \citep{osw2025}. This growth
occurred despite years of prior international dialogue on weapons
governance.

The escalation illustrates the dynamic our framework describes. Each
side faced overwhelming incentive to match the other's deployment
(C3-related competitive pressure under existential stakes). Drone
deployment by one side
created immediate pressure for counter-deployment (C2: negative
externalities). No international authority could enforce restraint
during active conflict (C5: no enforcement). We do not claim Condition
4 for the belligerents: in a war a mutual limit on drones may favor
one side, so the case shows competitive and enforcement pressure, not
a preference of both sides for mutual restraint. No agreement restraining
drone use existed between the parties, so the case illustrates
deployment pressure under existential stakes; it is not a test of
whether an agreement would have held. Pre-war governance discussions
had little effect once competitive pressure intensified.

Most drones in the war are flown remotely by a human pilot, and they
are physically observable, with detectable manufacturing and
deployment. The transition to AI-enabled autonomy has nevertheless
begun, and it follows the competitive logic described above. Electronic warfare is the immediate driver: by one account
about nine in ten strikes with exploding drones fail because the radio
signal is jammed or the target is beyond its range
\citep{kramer2026}, which creates an incentive to replace the radio
link by onboard navigation and terminal guidance. By late 2025 some
Ukrainian drones, once locked on a target, could pursue and strike it
with no further human involvement \citep{chivers2025}, and drone
makers on both sides were experimenting with systems that recognize
targets autonomously \citep{kramer2026}. Ukraine's Ministry of Defence
reports drones that can complete missions without a connection to the
operator, and names faster decisions and lower risk to personnel as
objectives \citep{modukraine2026a1}. In March 2026 it opened its
battlefield data to partners for the training of AI models, with the
stated aim of advancing the autonomy of drones and other combat
systems in the technological competition with Russia
\citep{modukraine2026data}; its defense minister said that ``we must
outperform Russia in every technological cycle'' and that ``the future
of warfare belongs to autonomous systems'' \citep{kramer2026}.

Autonomy in navigation and terminal guidance is distinct from
autonomous decisions to select and engage targets. Ukraine's stated
policy retains human control over final decisions
\citep{modukraine2026a1}, and weapons that identify, rank and pursue
several kinds of targets with no human action had not appeared by the
end of 2025 \citep{chivers2025}. The sources above name jamming and
the technology race as the drivers of autonomy. It is our conjecture
that the heavy casualties on both sides, more of them now inflicted by
drones than by any other weapon \citep{chivers2025, kramer2026}, the
difficulty of recruiting more soldiers, and the political cost of
deaths add to the pressure to replace soldiers with machines.

The same pressure extends beyond small drones. Jet-powered unmanned
combat aircraft are in development, for example Turkey's Bayraktar
K{\i}z{\i}lelma, whose manufacturer lists autonomous taxi, takeoff,
cruise and landing \citep{baykar2026}; we expect this development to
continue toward combat aircraft that fly and fight without pilots. As
more functions move into software, verifying agreed limits becomes
harder: it requires inspecting software, which adversaries have reason
not to share. We therefore expect the transition from
remote-controlled to autonomous systems to amplify the coordination
problem while adding severe verification difficulty.

If governance is absent for remotely piloted drones, a relatively
governable technology, we expect it to be harder still for autonomous
drones, and harder again as AI is integrated across other military
systems, from cyber operations and logistics to command and
control. The drone war thus serves as an illustration of the pressures
involved, not as a confirmation of the framework.

\subsection{Standard Solutions, Structural Barriers, and Possible Interventions}
\label{sec:ai-solutions}

Our framework suggests why commonly proposed AI governance solutions
face structural barriers, and it points to interventions that could in
principle address the AI governance tragedy by changing one or more of
the five conditions. The statements in this subsection are conditional
on the scope stated at the start of this section and are the author's
assessment. These are conditional mechanisms and assessments, not
proofs that the proposals fail; changing a condition is neither shown
necessary nor sufficient for successful coordination. We assess the
task as more difficult than for productivity because of the higher
assessed intensity of the conditions.

\paragraph{Standard proposals:} For each solution, the following table
notes which condition it does not by itself remove and the obstacle
this leaves.

{\setlength{\tabcolsep}{5pt}
\renewcommand{\theHtable}{proposals.ai}
\begin{longtable}{@{} Q{3.9cm} Q{12.1cm} @{}}
\toprule
\textbf{Proposal} & \textbf{Incentive or enforcement obstacle that remains} \\
\midrule
\endhead
\bottomrule
\endlastfoot
Voluntary industry self-regulation: AI labs voluntarily adopt safety
standards and slow development & The dominance condition (C3). In our
view an actor that restrains development while competitors advance
risks strategic elimination, which makes voluntary restraint
fragile. (Drawing lessons from nuclear arms control for military AI,
\citet{maas2019viable} argues by contrast that arms races are not
inevitable and can be managed, channeled or stopped.) The 2023 open
letter calling for a pause produced none despite thousands of
signatures \citep{fli2023pause}, and even safety-conscious
organizations face strong pressure to accelerate when rivals make
breakthroughs, as OpenAI's restructuring illustrates on our reading
(Section \ref{sec:ai-mapping}). \\
\addlinespace
International treaties: binding agreements limiting AI capabilities
development, analogous to nuclear non-proliferation treaties & The
enforcement obstacles (C5). AI faces severe verification challenges
(Table~\ref{tab:ai-comparison}): software can be developed covertly,
training can be concealed, and capabilities are hard to detect
reliably without access to proprietary systems. The Biological
Weapons Convention has no verification regime, and its effectiveness
has suffered for it \citep{koblentz2009}. No global authority can
enforce compliance. \\
\addlinespace
AI safety research: solve the technical alignment problem, making
advanced AI systems safe by design & In our assessment, insufficient
by itself even if technically successful. Safety research addresses
the risk of \textit{misaligned} AI, and can thereby lower the cost of
a race, but it does not by itself remove competitive pressure from
\textit{aligned} AI, which still confers economic and strategic
advantages on whoever is first. Whether the race that remains is worse
for every actor than coordinated restraint is the separate comparison
of Condition 4 (Section~\ref{sec:ai-mapping}); continued competition
alone does not show it. Labs may also cut corners on safety validation
when competitors approach breakthroughs \citep{armstrong2016}. \\
\addlinespace
National regulation: comprehensive safety requirements and capability
restrictions within individual nations & The enforcement obstacles
(C5), because of regulatory arbitrage. We argue that unilateral safety
regulations create competitive disadvantage and may cause development
to migrate to less cautious jurisdictions, shifting rather than
reducing existential risk; software-based development relocates faster
than physical industries. (The formal model of \citet{armstrong2016}
shows how competition can lead teams to skimp on safety; it does not
model jurisdictions.) \\
\addlinespace
Compute governance: restrict access to the advanced computing hardware
required for frontier AI development, as proposed by
\citet{anderljung2023} & Erosion over time. Algorithmic efficiency
improvements systematically lower hardware requirements, enabling
actors to achieve frontier capabilities with previously sub-frontier
hardware \citep{sastry2024}; and DeepSeek's release of smaller
distilled reasoning models in early 2025 illustrates how capabilities
can diffuse beyond the original developer \citep{deepseek2025}. Compute
governance also requires international cooperation, with the standard
enforcement barriers of sovereignty and verification. \\
\addlinespace
Energy monitoring: monitor data center energy consumption as a
verification mechanism, since large-scale training requires massive
electricity usage & At best an indirect signal of AI development
\citep{sastry2024}, with a dual-use problem. Nations invest in energy
infrastructure for reasons independent of AI (grid modernization,
electrification, industrial policy, general cloud computing demand),
so restricting it to limit AI development is economically and
politically infeasible, and large consumption could be attributed to
cryptocurrency mining, scientific computing or commercial cloud
services. Unlike specialized AI chips, energy infrastructure is
inherently general-purpose. \\
\end{longtable}
\addtocounter{table}{-1}
}

The analysis suggests that AI governance requires interventions of
unprecedented scope and coordination, beyond what successful arms
control achieved, because, on our assessment, AI development exhibits
the five tragedy conditions with high or extreme intensity.

\paragraph{Possible structural interventions:} Five interventions would each change one
of the conditions. We present them as possible structural intervention
targets, not policy recommendations, to clarify the magnitude of
barriers any solution must overcome. The following table states each
intervention and the difficulty it meets.

{\setlength{\tabcolsep}{5pt}
\renewcommand{\theHtable}{interventions.ai}
\begin{longtable}{@{} Q{2.5cm} Q{5.3cm} Q{7.85cm} @{}}
\toprule
\textbf{Condition to change} & \textbf{Intervention} & \textbf{Difficulty} \\
\midrule
\endhead
\bottomrule
\endlastfoot
Condition 1 (N-Player Structure) & Consolidate all transformative AI
development into a single global project, a heavily regulated monopoly
or international consortium. \citet{bostrom2014} terms this a
``Singleton'': a global coordination mechanism with sufficient
authority to prevent competitive AI development. Such an entity could
eliminate the racing dynamic by becoming the sole developer. &
Concentrating transformative AI development in a single entity enables
permanent power consolidation, potential totalitarian control over
transformative technology, and elimination of competitive safeguards
against misuse. The cure might be worse than the disease, trading
coordination tragedy for concentration risk. Achieving global
agreement to establish such an entity also faces formidable obstacles
given competing national interests, sovereignty concerns, and deep
mistrust among major powers regarding who would control it. \\
\addlinespace
Condition 2 (Binary Choice with Externalities) & Develop AI systems
that inherently do not create strategic advantages, eliminating
competitive pressure. & This is structurally difficult because AI
capabilities \textit{inherently} create strategic advantages: economic
productivity, military capability, or scientific advancement cannot be
developed in ways that do not affect competitive positions. This
reflects the \citet{collingridge1980} dilemma: the strategic
implications and competitive externalities of transformative AI only
become fully apparent once the technology is developed, at which point
competitive dependencies impede coordinated restraint. Early restraint
tends to fail from uncertainty about implications; late restraint
tends to fail from path dependency and lock-in to competitive
dynamics. \\
\addlinespace
Condition 3 (Dominance Property) & Reduce winner-take-all dynamics and
existential stakes making defection dominant. Three mechanisms have
been proposed: (1) credible international agreements reducing security
competition between major powers; (2) economic restructuring through
profit-sharing mechanisms such as the `Windfall Clause' proposed by
\citet{okeefe2020}, where AI developers pre-commit to distributing
exceptional profits (exceeding 1\% of world GDP) broadly; (3)
demonstrated proof that advanced AI poses catastrophic risks even to
prospective winners, undermining the perceived advantage of racing. &
Each faces a substantial barrier. (1) Such agreements face severe
commitment credibility problems: nations cannot reliably forego
military advantages from transformative AI when adversaries might
defect. (2) Profit-sharing requires binding pre-commitment before
competitive pressures emerge, making voluntary adoption difficult, and
cannot address non-economic strategic advantages: military capability
and geopolitical influence remain winner-take-all even with
profit-sharing. (3) Establishing convincing evidence of unprecedented
catastrophic risks before deployment is epistemically difficult;
waiting for direct evidence of catastrophe would be too late. As \citet{bostrom2014} notes, the difficulty of credibly
demonstrating novel existential risks before they materialize is a
fundamental challenge, and even acknowledged risks may be insufficient
when individual actors face dominant incentives to proceed if they
believe others will do so regardless. \\
\addlinespace
Condition 4 (Pareto-Inefficiency) & Alternative assessment: accept
the racing equilibrium as optimal, rejecting the premise that coordination would improve
outcomes. Some argue that rapid AI deployment's benefits outweigh
coordination costs, effectively denying the Pareto-inefficiency
condition. & The concern is that existential risks from misaligned
superintelligence are non-excludable \citep{bostrom2014,
  russell2019}. Even prospective winners may therefore prefer
coordinated restraint. Whether they do depends on each actor's
comparison of the benefits, costs and risks of both outcomes, as
discussed in Section~\ref{sec:ai-mapping}; non-negligible existential
risk alone does not establish Condition 4. \\
\addlinespace
Condition 5 (Enforcement Difficulty) & Implement hardware-based
verification using advanced semiconductor manufacturing as a
chokepoint. \citet{anderljung2023} argue that `compute governance',
e.g., monitoring and controlling access to high-end GPU clusters
required for frontier AI training, offers the most viable enforcement
mechanism, as specialized hardware is more detectable than software
development. A global regime tracking chip production, data center
construction, and computing cluster deployment could theoretically
enforce development limits. & This advantage faces the erosion
described above for compute governance: over multi-year timescales,
efficiency improvements enable powerful AI development on widely
available hardware. Effective governance also requires unprecedented
international cooperation among competing powers, including
adversarial nations with strong incentives to defect for strategic
advantage, and verification of compliance faces challenges from
concealed computing clusters and clandestine development programs. \\
\end{longtable}
\addtocounter{table}{-1}
}

\paragraph{Comprehensive Solution:} In our assessment, a practical
response to the AI governance tragedy is likely to combine
interventions addressing several conditions: an
international consortium among leading AI powers (partial C1
consolidation), compute governance with hardware tracking (partial C5
enforcement), windfall clauses and benefit-sharing mechanisms (partial
C3 reduction), and continued investment in technical AI safety (risk
mitigation if governance fails). However, each intervention faces
severe implementation barriers, and we expect partial solutions to
leave the tragedy substantially intact. Unlike nuclear weapons, where
physical constraints enabled partial success through non-proliferation
regimes, AI's software nature, dual-use characteristics, and low entry
barriers for replication and deployment create coordination
requirements beyond historical precedent. On the dimensions assessed
here, AI governance appears structurally harder than the historical
arms-control cases compared above.

\subsection{Reconciliation with Conservative Productivity Estimates}

\citet{acemoglu2024simple} provides conservative estimates of AI's
productivity impact: in the April 2024 working-paper version, an upper
bound of about 0.55\% on cumulative total factor productivity gains
over ten years after adjusting for hard-to-learn tasks (0.71\% before
that adjustment). In an October 2026 interview, Acemoglu acknowledged
that agentic AI had arrived faster than he anticipated when writing
his earlier productivity analysis, while maintaining that adoption in
production remained slow \citep{acemoglu2026interview}. His 2024
estimate should therefore be read in the context of its assumptions
and the information available at the time.
However, we argue that the incentive structure does not depend on
large productivity effects. Even if AI offers only modest advantages,
the dominance condition (C3) may hold: firms may be unable to forego
even marginal efficiencies when competitors adopt them. This concerns
incremental adoption by firms, where the gain must also cover any cost
of adoption that a firm avoids by not adopting; the transformative
scenario assessed earlier in this section is a different regime. A gap between
growth rates widens over time, so economic and military advantages
compound. Whether dominance holds at a given magnitude depends on the
payoffs.

Uncertainty may amplify this dynamic. Consistent with Prospect Theory
\citep{kahneman1979prospect}, actors facing the perceived certainty
of strategic obsolescence may gamble on risky acceleration rather
than accept certain disadvantage.

\section{Implications and Falsifiability}\label{sec:implications}

Having argued that both productivity competition and AI development
exhibit the five conditions for structural tragedy, we now state what
the framework predicts, what would count against it, and its scope and
limitations.

\subsection{What the Framework Predicts}

The two applications generate conditional expectations that can be
compared with empirical evidence. Each combines the stage-game
condition of the setting concerned (condition \eqref{eq:Ei}, holding
strictly, for productivity; Conditions 3 and 4 for AI) with the
persistence hypothesis of Section \ref{sec:formal-result}:

\paragraph{For Productivity:} (1) Gains in productivity flow mainly to
output and not to shorter hours, so that working hours do not decline
substantially in competitive market economies absent structural
intervention. (2) Restraint that relies on voluntary agreement or on a
cultural shift toward leisure, with incentives unchanged, does not
last. (3) National regulations reducing hours remain under competitive
pressure and erosion.

\paragraph{For AI:} (1) Voluntary pauses on capabilities development
fail to materialize or collapse quickly, and safety commitments erode
under competitive pressure as capabilities advance. (2) Racing
intensifies as capabilities approach human-level performance. (3)
National regulation alone is circumvented, and international
coordination faces structural barriers where enforcement is missing.

\paragraph{Falsifiability:}

Theorem \ref{thm:main} and the propositions are mathematical
statements and are not what these observations test. Three kinds of
observation would count against the applications.

\emph{Against the stage-game analysis} (it would show that the
incentives are not as
modeled): evidence that a firm which restrains alone while its rivals
expand earns more than it would by expanding; or estimates of AI
developers' risk-adjusted payoffs that fail the stated inequalities.

\emph{Against the application and its persistence hypothesis}
(it would show that cooperation can arise or last with incentives,
information and
institutions unchanged): working hours declining substantially despite
continued competition and without structural intervention such as
regulation or collective bargaining; a sustained and verifiable
voluntary pause on AI capabilities development lasting 12+ months with
major actors participating; or a lasting voluntary agreement among
competing firms to pass productivity gains into shorter hours.

\emph{Against our assessment of the enforcement obstacles} (it would
show that enforcement is less difficult than we have rated it): a
verifiable global cap on compute or
capabilities that is enforced without defection, or binding
international labor standards that are monitored and enforced.

The third kind would revise our assessment of how hard the obstacles
are to overcome; it would not contradict the stage-game results, which
describe the situation before such an institution exists. Changing one
of the five conditions, as deposit insurance did for bank runs, is
consistent with the framework; whether a particular change succeeds
requires analyzing the incentives and institutions that result.

These predictions are specific enough to be empirically testable. The
framework does not claim that coordination is impossible, only that it
is structurally difficult: sustained cooperation without a change in
incentives, information or institutions would count against our
analysis of the application concerned.

\subsection{Scope and Limitations}

This framework identifies structural conditions that create
coordination difficulties, but we do not claim to provide
comprehensive causal explanations for all historical variation in
productivity-welfare relationships or coordination outcomes across
domains.

\paragraph{Modeling Limitations:} The productivity model omits buyers
and treats hours as the firm's choice, with fixed pay; its results
concern the modeled firms and their workers, not total welfare. For AI
development, Conditions 3 and 4 are argued informally and are not
derived from a model; whether the race is a dilemma or an assurance
problem is open, and Section \ref{sec:ai} states the inequalities that
separate the two. For Condition 5, Proposition \ref{prop:trigger}
treats a specified trigger arrangement and, when a defector cannot be
held below $P_i$, characterizes permanent full cooperation under two
monitoring structures; false alarms, outside sanctions and general
institutional design remain open. Persistence over time in the
applications is a hypothesis.

\paragraph{Variation over time, across countries and in intensity:}
The framework does not by itself explain historical variation. The
intensification of the tragedy after 1980 coincides with
globalization, union decline and increased capital mobility, which
strengthen the conditions (Section \ref{sec:uses}); establishing the
causal mechanisms requires detailed historical analysis. Across
countries the tragedy appears in different dimensions (Section
\ref{sec:europe}): the framework predicts that productivity gains will
not flow fully to workers when competitive conditions are strong, not
which dimension bears the adjustment in each institutional
context. Where conditions hold weakly, partial coordination may be
more feasible and sustainable, as lower European hours suggest; how
condition intensity relates to the degree of coordination difficulty
and to the sustainability of partial coordination should be
investigated systematically.

\paragraph{Empirical Validation:} We offer suggestive evidence and
illustrations: the declining labor share of income, productivity-wage
decoupling and the European case for productivity, and the
Russia-Ukraine drone war for AI development. However, comprehensive testing across historical
periods, countries, sectors, and coordination attempts is necessary to
fully validate our framework and identify boundary conditions. We view
this paper as establishing a diagnostic framework and providing
suggestive evidence, with systematic empirical validation remaining an
important direction for future research.

\section{Conclusion}\label{sec:conclusion}

\subsection{Summary of Contributions}

We have made four contributions. First, a diagnostic framework of five
conditions, synthesized from established game theory, of which
dominance and Pareto-inefficiency are sufficient for a unique,
Pareto-dominated equilibrium of the stage game and, among identical
actors, necessary for it. Second, a heuristic Tragedy Index of graded
condition intensities, under which governance of transformative AI
breakthroughs receives the highest composite score among the cases
assessed, above climate change, nuclear weapons, and bank
runs. Third, building on \citet{schnaiberg1980}, a stylized model of
productivity competition with an exact condition on each firm's hours
and market share under which universal output expansion is the unique,
Pareto-dominated equilibrium, which offers an explanation of why
Keynes's expectation has not been realized; the European case is
consistent with the hypothesis that the tragedy is mitigated at high
cost and partly redistributed. Fourth, the argument that AI
development has the same structure, conditionally and with higher
assessed intensity, and that its governance compares unfavorably with
historical arms control; the drone war illustrates the dynamic.

\subsection{Future Research}

Several problems are open. A general characterization of level 2 for
asymmetric binary games, beyond the sufficient test of Conditions 3
and 4 and the exact result for the productivity model, is one.
Proposition \ref{prop:trigger} treats a specified trigger arrangement
and, when a defector cannot be held below $P_i$, characterizes
permanent full cooperation under two monitoring structures; false
alarms, outside sanctions and general institutional design remain
open. For AI development Conditions 3 and 4 are argued, not derived: a
payoff model that separates a dilemma from an assurance problem is
needed. The productivity model has two actions, fixed revenue and a
common productivity factor; continuous hours, demand that responds to
price, bargaining over pay, risk aversion and persistent differences
in productivity gains would each test whether its condition survives.

On the empirical side, the Tragedy Index is a rubric: whether graded
intensity tracks observed governance outcomes, and whether
payoff-based measures such as the Price of Anarchy rank cases the same
way, can be examined on historical cases such as the Montreal Protocol
and the nuclear test bans, where one can also ask whether a condition
was changed or only overridden, and which conditions most resist
intervention. What makes a costly institutional override last, as in
Europe, is a related question. The checklist can be applied to new
domains (quantum computing, synthetic biology, geoengineering) and to
arms races, trade wars, regulatory competition and financial
contagion. Finally, the two tragedies may interact, for example if
productivity competition supplies resources and pressure for the AI
race, and military uses of AI may form a tragedy of their own if
autonomous systems lower the threshold for starting a conflict.

\bibliographystyle{plainnat}
\bibliography{dasdan}

\end{document}